\documentclass[aps,twocolumn,floatfix,nofootinbib]{revtex4-2}

\usepackage{amsmath,amsthm,amsfonts,amssymb,times,bbm,graphicx,bbold,mathrsfs}
\usepackage[hidelinks, colorlinks=true, allcolors=blue]{hyperref}
\usepackage{float}
\graphicspath{{./images/}}

\usepackage{multirow}

\newtheorem{theorem}{Theorem}

\newtheorem*{proposition*}{Proposition}
\newtheorem{lemma}[theorem]{Lemma}
\newtheorem*{lemma*}{Lemma}
\newtheorem{corollary}[theorem]{Corollary}
\newtheorem{definition}[theorem]{Definition}

\newtheorem{remark}[theorem]{Remark}

\newtheorem{condition}{Condition}

\newcommand{\RR}{\mathbb{R}}

\newcommand{\ZZ}{\mathbb{Z}}
\newcommand{\NN}{\mathbb{N}}
\newcommand{\CC}{\mathbb{C}}

\renewcommand{\S}{\mathcal{S}}
\newcommand{\E}{\mathcal{E}}
\renewcommand{\P}{\mathcal{P}}
\newcommand{\D}{\mathcal{D}}

\renewcommand{\L}{\mathcal{L}}

\newcommand{\Z}{\mathcal{Z}}

\newcommand{\K}{\mathcal{K}}
\newcommand{\T}{\mathcal{T}}

\newcommand{\Em}{\mathscr M}

\newcommand{\Cee}{\mathscr C}
\newcommand{\Gee}{\mathscr G}

\newcommand{\Ess}{\mathscr S}

\newcommand{\Oh}{\mathscr O}

\newcommand{\1}{\underline{1}}
\newcommand{\Id}{\mathbb{1}}

\DeclareMathOperator{\tr}{tr}
\DeclareMathOperator{\Span}{span}

\DeclareMathOperator{\img}{img}

\DeclareMathOperator{\len}{len}

\usepackage{xcolor}

\usepackage{algorithm}
\usepackage{algorithmicx}
\usepackage{algpseudocode}
\algrenewcommand\algorithmiccomment[1]{\hfill {#1}}

\usepackage{stmaryrd}
\newcommand{\swap}[3]{\left\llbracket {}_{#1}{}^{#2}{}_{#3} \right\rrbracket}

\begin{document}

% \nocite{*}

% \title{Self testing the sphere}

% Alternate title

\title{An operational characterization of finite-dimensional quantum theory}

\author{Lionel J.\ Dmello${}^1$}
\email{ldmello@thp.uni-koeln.de} 
\author{Xiangling Xu${}^2$}
\author{Marc-Olivier Renou${}^2$}
\author{David Gross${}^1$}

\affiliation{${}^{1}$Institute for Theoretical Physics, University of Cologne, Germany}
\affiliation{${}^{2}$Inria, CPHT, LIX, CNRS, \'Ecole Polytechnique, Institut Polytechnique de Paris, Palaiseau, France}

\date{Sept 21, 2026} 

\begin{abstract}
    A key goal in the foundations of quantum mechanics is to identify operational constraints characterizing physical theories. 
    Bell inequalities do so for classical probability theory, while Tsirelson's bound provides a first step for quantum mechanics.
    Here, we address the dual question:
    Can one certify that all correlations predicted by quantum theory are actually realizable?
    In this work, we construct a finite number of two-body correlations such that the only probabilistic theory that
    (1) realizes them, and
    (2) does so in a way that is stable under iterated teleportation,
    is quantum theory.
    For $n$ quantum systems of dimension $d$ each, Condition~(1) can be verified using $\operatorname{poly}(d,n)$ measurement settings.
    Condition~(2) may be understood as a hierarchy of tests, one for each number $N$ of teleportation steps.
    There is thus a sense in which finite-dimensional quantum theory can be self-tested.
    In particular, one can certify the existence of Bell inequality violations larger than any that have been directly observed.
\end{abstract}

\maketitle 

%================================================================================================
% Introduction
%================================================================================================

\section{Introduction}\label{sec:introduction}

A central goal in the foundations of quantum mechanics is to understand what singles out quantum theory among all conceivable physical theories. 
Historically, this has been studied by postulating physical principles like no-signalling or information causality \cite{winter2009informationcausality, miklin2021infocausality} (or many others \cite{fritz2013localorth, hardy2001quantumtheoryreasonableaxioms, chiribella2017quantum, wilce2018royalroad, brassard2006limit, linden2007quantum, navascues2010glance,clauser1969proposed})
that constrain the set of allowed correlations ``from the outside''.
That is, the usual approach is to provide an explanation for the absence of super-quantum correlations in Nature.

Here, we address the dual question:
Can one certify from observations that all states and effects predicted by quantum theory are actually realizable?
We answer this question for finite-dimensional systems in the affirmative.

For $n$ quantum systems of dimension $d$ each, we construct a two-step protocol. 
The first step is a three-party entanglement swapping experiment. 
The correlations observed in this experiment encode properties of a dynamical map 
in a way that is analogous to the well-known quantum mechanical \emph{gate teleportation protocol}.
The second step consists of a hierarchy of tests,
which at level $N\in\NN$ checks whether these correlations are preserved under $N$ iterated entanglement swapping steps.
We show that if a theory passes every test in this hierarchy,
then the analogy with quantum gate teleportation becomes exact,
and the correlations uniquely specify a dynamical map. 
If the maps characterized this way generate a group isomorphic to the symmetry group $U(d^n)$ of $n$ qudits, 
it follows that the theory must contain all states and effects realizable in quantum mechanics.

The technical argument relies on the analysis of iterated entanglement swapping experiments in probabilistic theories, a method previously developed in Refs.~\cite{weilenmann2020self, mirjam2020achsh, dmello2024entanglement, dmello2026probabilistic}.
The main limitation of our result is that, unlike in Bell's theorem, no single experiment rejects all alternative theories: 
quantum mechanics is characterized as the unique theory that passes a hierarchy of tests labeled by a \emph{level} $N\in \NN$.

%========================================================
% General probabilistic theories
%========================================================

\subsection{General Probabilistic Theories (GPTs)}\label{subsec:GPTs}

In this section we introduce the GPT formalism from our (LD and DG) previous work \cite{dmello2026probabilistic}, which was itself adopted from even earlier work \cite{dmello2024entanglement}.
Here we provide an overview of the formalism and the modifications.
For a more detailed description, please refer to Sec.~I and Sec.~II of \cite{dmello2024entanglement}.
For a general overview of the GPT formalism we refer the reader to \cite{plavala2023general}.

\emph{General probabilistic theories} (GPTs) \cite{segal1947postulates,ludwig1967attempt2,ludwig1968attempt3,dahn1968attempt,stolz1969attempt,stolz1971attempt,davies1970operational,barrett2007information, plavala2023general} is a framework that allows us to describe correlations that are compatible with operationally motivated assumptions (e.g.\ no-signaling), but are not necessarily realized in QM.
The goal of such a framework is to study the observations of QM in a broader setting in order to extract the underlying physical principles governing them.

We assume that a GPT is specified by 
\begin{itemize}
    \item A finite-dimensional real vector space $V$.
    \item A collection of closed, pointed convex cones\footnote{A convex cone is said to be closed if it contains all it's limit points with respect to some topology. 
    In the present case, since we work with finite dimensional real vector spaces, we take the standard topology. 
    A convex cone $P$ is said to be pointed if $P \cap -P = \{ 0 \}$.},
    \begin{align*}
        \forall n \in \NN, \P^{(n)} \subseteq V^{\otimes n},
    \end{align*}
    where $V^{\otimes n}$ is the standard $n$-fold tensor product of $V$.
    \item A distinguished element $\1$ in the relative interior\footnote{The relative interior of a convex cone is the interior of the cone with the ambient space taken to be the affine hull of the cone.} of $\P^{(1)}$.
\end{itemize}
We call $\P^{(n)}$ the \emph{effect cones} and $\1$ the \emph{unit effect}.
The effects model measurement outcomes, in the sense that the effect $e$ encodes the probability of obtaining the outcome labelled by $e$.
The unit effect in particular, models the trivial measurement, which encodes the fact that in every experiment some outcome must occur.
Therefore, a set of effects $\{e_i\}_{i = 1}^m \subseteq \P^{(n)}$ defines a measurement if, 
\begin{align*}
    \sum_{i = 1}^m e_i = \1^{\otimes n}. 
\end{align*}

In this formalism, \emph{states} are viewed as linear functionals on effects.
To this end let $(\P^{(n)})' \subseteq (V^{\otimes n})^*$ be the polar dual\footnote{The polar dual $P'$ of a convex cone $P$ is the set of all non-negative linear functionals on $P$.} of $\P^{(n)}$.
We choose a set of pointed cones $\forall n \in \NN : \D^{(n)} \subseteq (\P^{(n)})'$, that are generating\footnote{A convex cone is said to be generating if it spans the ambient space.} on $\Span(\P^{(n)})^*$, called the \emph{state cones}.
The convex sets
\begin{align*}
    \S^{(n)} := \{ \rho \in \D^{(n)} \ | \ \rho(\1^{\otimes n}) = 1 \},
\end{align*}
are called the \emph{state spaces} of the GPT.
The state spaces are compact sets (App.~\ref{app:compact_state_sapce}), a fact that will be important later (Sec.~\ref{subsec:comp_topo_semigrp}).

Let $e \in \P^{(n)}$, define the \emph{negation} of $e$ as 
\begin{align*}
    \neg e := \1^{\otimes n} - e.
\end{align*}
The (convex) set of all those elements of $\P^{(n)}$ whose negation is also in $\P^{(n)}$ is called the \emph{effect space} $\E^{(n)}$, i.e., 
\begin{align*}
    \E^{(n)} := \P^{(n)} \cap \neg \P^{(n)}.
\end{align*}
This definition implies that 
\begin{align*}
    e \in \E^{(n)} \Leftrightarrow e \in \P^{(n)} \ \text{and} \ \neg e \in \P^{(n)},
\end{align*}
from which it directly follows that 
\begin{align*}
    e \in \E^{(n)} \Leftrightarrow \forall \rho \in \S^{(n)}, 0 \leq \rho(e) \leq 1.
\end{align*}
That is, the pairing between any effect $e$ and state $\rho$ is not only positive, but also bounded above by $1$.
As such it will be interpreted as the probability of obtaining the outcome corresponding $e$ given we prepared the state $\rho$.
Additionally, every element of $\P^{(n)}$ with this property is in $\E^{(n)}$.

%========================================================
% GPT closure
%========================================================

\subsubsection{Constructing a GPT from a set of states and effects}\label{subsubsec:generating_a_GPT}

It is often the case that we seek to determine whether it is \emph{possible} to achieve certain correlations in some experimental scenario.
Addressing this question typically only requires us to specify states and measurements which, if employed in the experiment, reproduce the correlations in question.
Therefore, it is convenient to have an algorithmic procedure that takes a collection of states and effects and turns it into a GPT that contains them.

Algorithm~1 of \cite{dmello2024entanglement} provides a straight-forward ``closure construction'' that addresses exactly this problem.
In this work, in contrast to the treatment in \cite{dmello2024entanglement}, we do no assume that GPTs are invariant under permutation of subsystems.
Thus, when we refer to the \emph{GPT closure} of sets of states and effects, we mean the output of Algorithm~1 of \cite{dmello2024entanglement} with the symmetrization step omitted.

In this work we focus mainly on entanglement swapping experiments, which can be stated by combining only bipartite states and effects.
For this reason, we will typically consider only GPTs generated by taking the GPT closure of bipartite objects.

%========================================================
% Outline
%========================================================

\subsection{Outline}\label{subsec:outline}

In order to have a clear exposition of the ideas involved, we first describe the entire scheme for the simplest case, namely the qubit.
Then, in the later part of the manuscript, we will show how to generalize from one to many qubits, as well as qubits to qudits.
To this end, in Sec.~\ref{sec:group_gen_expt_qubit} we describe the entanglement swapping experiment which encodes the symmetries of the qubit state space.
In Sec.~\ref{sec:teleportation_heirarchy_qubit}, assuming we have access to independent
%and identically distributed (IID) 
copies of a bipartite preparation $\rho$, we describe the hierarchy of tests that check for stability under teleportation, and subsequently derive the constraints implied by it.
In Sec.~\ref{sec:obtaining_the_qubit} we prove the main claim for the qubit.
In Sec.~\ref{sec:relaxing_iid}, we will show that the IID assumption incurs no loss of generality.
And finally, in Sec.~\ref{sec:fin_dim_qm}, we generalize the hierarchy to $n$ qudits.

%================================================================================================
% GROUP GENERATOR EXPERIMENT
%================================================================================================

\section{The group generator experiment for the qubit}\label{sec:group_gen_expt_qubit}

The \emph{group generator experiment} is an entanglement swapping experiment performed by three parties.
Let the three parties be called Alice, Bob, and Erwin.
The experimental setup is described in Fig.~\ref{fig:gg_test}.
Alice and Bob have a three-settings two-outcome measurement machine, whereas Erwin has a three-setting four-outcome measurements machine.
A bipartite preparation is shared between Alice and Erwin, and Erwin and Bob.

\begin{figure}[H]
    \centering
    \includegraphics[trim={0cm 4cm 0cm 4cm}, clip, width = \linewidth, keepaspectratio]{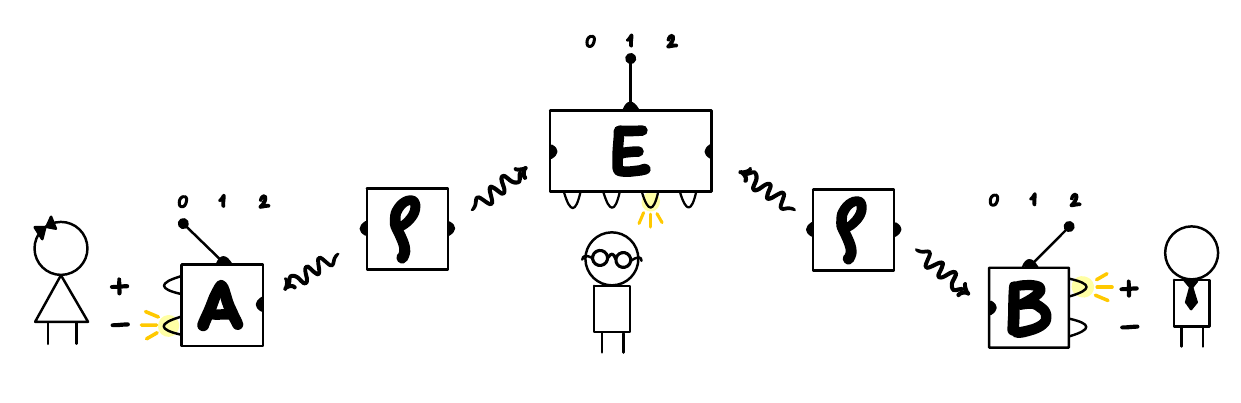}
    \caption{The entanglement swapping experiment for the qubit.}
    \label{fig:gg_test}
\end{figure}

In each round of the experiment, a referee sends questions $(s, r, t) \in \ZZ_3 \times \ZZ_3 \times \ZZ_3$ to Alice, Erwin, and Bob, respectively.
Alice and Bob have to reply with $(a,b) \in \ZZ_2 \times \ZZ_2$. 
In other words, we can model each setting of Alice and Bob with $\pm1$ valued observables $A_s, B_t$ respectively.
Erwin on the other hand as to provide a two-bit (four-outcome) response $\mu \in \ZZ_2^2$.
Let $\phi^{(r)}_{\mu}$ denote Erwin's measurement effect for setting $r$ and outcome $\mu$.
The ``target correlation'' of the experiment is based on the following quantum realization of the experiment.

%========================================================
% Quantum realization
%========================================================

\subsection{Quantum realization}\label{subsec:quantum_realization_for_qubits}

Let the state shared between Alice and Erwin, and Erwin and Bob, be the Bell state
\begin{align*}
    |\Phi^+\rangle = \tfrac1{\sqrt 2} (| 00 \rangle + | 11 \rangle).
\end{align*}
Denote the settings of Alice and Bob by $s, t \in \ZZ_3$.
Using the usual notation for the Pauli operators -- $\sigma_{00} = \Id, \sigma_{01} = \sigma_x, \sigma_{10} = \sigma_y, \sigma_{11} = \sigma_z$, $\pm1$ valued observables of Alice and Bob are 
\begin{itemize}
    \item 
    $A_0 = B_0 = \sigma_{01} = \sigma_x$,
    \item
    $A_1 = B_1 = \sigma_{10} = \sigma_y$,
    \item
    $A_2 = B_2 = \sigma_{11} = \sigma_z$.
\end{itemize}

Let $R^z_{\theta}$ and $R^x_{\theta}$ denote the three-dimensional rotations by angle $\theta$ about the $z$ and $x$-axes respectively.
That is, 
\begin{align*}
    R^z_\theta
    :=
    \begin{pmatrix}
        \cos\theta & \sin\theta & 0\\
        -\sin\theta & \cos\theta & 0\\
        0 & 0 & 1\\
    \end{pmatrix},
    \quad 
    R^x_\theta
    :=
    \begin{pmatrix}
        1 & 0 & 0\\
        0 & \cos\theta & \sin\theta\\
        0 & -\sin\theta & \cos\theta\\
    \end{pmatrix}.
\end{align*}
For any $\theta \in (0, \pi/4]$, these matrices generate a dense subset of the group $SO(3)$ (this is a special case of a much more general fact \cite{breuillard2002densefreesubgroupslie}).
Let $\{ \xi_r \}_{r \in \ZZ_3} \subseteq U(2)$ be such that  
\begin{itemize}
    \item $\xi_0 = \Id$,
    \item $\xi_1(\cdot)\xi_1^\dag$ implements $R^z_{\theta}$ on the Bloch sphere, 
    \item $\xi_2(\cdot)\xi_2^\dag$ implements $R^x_{\theta}$ on the Bloch sphere.
\end{itemize}
For each $r \in \ZZ_3$ Erwin performs the following basis measurement
\begin{align*}
    \{ \sigma_\mu\xi_r^\dag \otimes \Id |\Phi^+\rangle \ | \ \mu \in \ZZ_2^2 \}.
\end{align*}
That is, for setting $r = 0$, it is the usual Bell-basis measurement.
For settings $r = 1,2$, we apply the local rotation $\xi_r$ before performing the Bell-basis measurement.

We will express the correlations between Alice and Bob, conditioned on Erwin choosing setting $r$ and obtaining outcome $\mu$, in the form of a matrix ($\1 = \Id$):
\begin{align*}
    \Cee^{(r)}_\mu 
    :=
    \begin{array}{c|cccc}
        & B_0 & B_1 & B_2 & \1 \\
        \hline
        A_0 & \langle A_0B_0\rangle & \langle A_0B_1\rangle & \langle A_0B_2\rangle & 0 \\
        A_1 & \langle A_1B_0\rangle & \langle A_1B_1\rangle & \langle A_1B_2\rangle & 0 \\
        A_2 & \langle A_2B_0\rangle & \langle A_2B_1\rangle & \langle A_2B_2\rangle & 0 \\ 
        \1 & 0 & 0 & 0 & 1 
    \end{array}
\end{align*}
To simplify formulas we will use the alias $A_3 = B_3 = \1$.
Define 
\begin{gather*}
    \K_{00} := \Id_4, \\
    \K_{01} := 
    \begin{pmatrix}
        1 & 0 & 0 & 0 \\
        0 & -1 & 0 & 0 \\
        0 & 0 & -1 & 0 \\
        0 & 0 & 0 & 1 \\
    \end{pmatrix},
    \quad
    \K_{10} := 
    \begin{pmatrix}
        -1 & 0 & 0 & 0 \\
        0 & 1 & 0 & 0 \\
        0 & 0 & -1 & 0 \\
        0 & 0 & 0 & 1 \\
    \end{pmatrix},
    \\
    \K_{11} := \K_{01}\K_{10}.
\end{gather*}
and
\begin{gather*}
    \Em^{(0)}
    :=
    \begin{pmatrix}
        1 & 0 & 0 & 0 \\
        0 & -1 & 0 & 0 \\
        0 & 0 & 1 & 0 \\
        0 & 0 & 0 & 1 \\
    \end{pmatrix},
    \\
    \Em^{(1)}
    :=
    \begin{pmatrix}
        R^z_\theta & 0 \\
        0 & 1 \\
    \end{pmatrix}
    \Em^{(0)},
    \quad
    \Em^{(2)}
    :=
    \begin{pmatrix}
        R^x_\theta & 0 \\
        0 & 1 \\
    \end{pmatrix}
    \Em^{(0)}.
\end{gather*}
Then, for the aforementioned choices of states and measurements, we get the correlations:
\begin{equation}\label{eqn:Corr_mat_def}
    \Cee^{(r)}_\mu = \K_\mu \Em^{(r)}.
\end{equation}
For the derivation of the above correlation matrices refer to App.~\ref{app:qm_corr_mats}. 

The correlation matrices $\Cee^{(r)}_\mu$ will be used to define the target correlation in the teleportation experiment discussed in the next section.

%================================================================================================
% THE TELEPORTATION HIERARCHY
%================================================================================================

\section{The teleportation hierarchy for the qubit}\label{sec:teleportation_heirarchy_qubit}

In this section we describe a hierarchy of entanglement swapping experiments whose statistics encode the fact that 
the correlations of the group generator experiment (Sec.~\ref{sec:group_gen_expt_qubit}) are stable under iterated teleportation.
Subsequently, we deduce the constraints on the GPT due to this condition.

Just as in the case of self-testing \cite{vsupic2020self}, it is not possible to make states about the parts of the GPT that are ``unreachable'' given some starting data.
Thus, when we say that we deduce constraints on the GPT, we mean the GPT closure (Sec.~\ref{subsubsec:generating_a_GPT}) of the states and measurements that go into the specification of the experiment.
We call this the \emph{effective GPT}.
Keeping this in mind, later we will show that the only (effective) GPT satisfying these conditions is the qubit GPT.

Before we begin we re-emphasize that, in what follows, we do not assume the postulates of QM to hold.

%========================================================
% Notation
%========================================================

\subsection{Some notation}\label{subsec:notation}

In the succeeding parts, we will consider the case where there are multiple entanglement swappers, i.e., multiple ``Erwins''.
For the case of $N$ Erwins, we will collect the settings of all the Erwins in a list denoted by $\vec r \in \ZZ_3^N$.
Correspondingly, we will have a list of outcomes $\vec \mu \in (\ZZ_2^2)^N$.
In fact we will restrict our attention to a special subset of settings, namely those where all but one Erwin is chooses the setting $0$.
To this end, define the set of \emph{``zero-padded'' setting vectors} as follows 
\begin{equation}\label{eqn:zero_padded_vects}
    \Z := \{ (\underbrace{0, \dots, 0}_{n_1}, r, \underbrace{0, \dots, 0}_{n_2}) \ | \ n_1, n_2 \in \NN,\ r \in \ZZ_3 \}.
\end{equation}
We will use the notation $\vec 0_{n_1} \circ r \circ \vec 0_{n_2}$ (``$\circ$'' here denoted concatenation of lists) to denote elements of $\Z$.
This sort of partitioning of the setting vector gives a similar partitioning of the outcome vectors which we will denote by $\vec m_1 \circ \mu \circ \vec m_2$, where $\mu \in \ZZ_2^2$, and $\vec m_1, \vec m_2,$ are vectors of appropriate length with entries in $\ZZ_2^2$.

In the quantum realization of the experiment, the teleportation protocol preserving the group generator correlations corresponds to the usual teleportation protocol. 
It is comprised of a Bell-state and a Bell-basis measurement.
The by-product/correction group of this protocol is $\ZZ_2^2$ under bit-wise addition modulo two.
The matrices $\K_\mu$ (Sec.~\ref{subsec:quantum_realization_for_qubits}) encode this fact and thus form a representation of $\ZZ_2^2$ (in fact, the regular representation), for which we have the homomorphism
\begin{align*}
    \K_\mu \K_\nu = \K_{\mu \oplus \nu},
\end{align*}
where ``$\oplus$'' denotes bitwise X-OR.
Given an outcome vector $\vec \mu = (\mu_1, \dots, \mu_N) \in (\ZZ_2^2)^N$ we will use the notation 
\begin{align*}
    \prod_{i = 1}^N \K_{\mu_i} = \K_{\vartheta(\vec \mu)},
    \qquad
    \vartheta(\vec \mu) := \bigoplus_{i = 1}^N \mu_i.
\end{align*}

%========================================================
% The hierarchy
%========================================================

\subsection{The hierarchy}\label{subsec:the_hierarchy}

Consider a bipartite preparation procedure described by a state $\rho$.
Here, we assume that we have access to an arbitrary number of independent and identical (IID) copies 
of $\rho$.
% Classically, the "D" in IID refers to "distributed", which doesn't really make sense for GPT states. Keeping it anyway for now. (It could refer to IDentical. :))

\begin{definition}[The teleportation hierarchy $\T_N$]\label{def:teleportation_hierarchy}
    Let $N \in \NN$.
    Given $N+1$ copies of a bipartite preparation $\rho$, where the left subsystem of the $k^{\mathrm{th}}$ copy is labelled $2k - 1$ and the right subsystem is labelled $2k$, the \emph{$N^{\mathrm{th}}$ level of the teleportation hierarchy, $\T_N$,} is an entanglement swapping experiment that is performed as follows:
    \begin{itemize}
        \item A referee sends subsystem $1$ to Alice, $2N+2$ to Bob, and subsystems $(2k, 2k+1)$ to Erwin${}_k$ (Fig.~\ref{fig:Test_TN}).

        \item The referee sends questions $(s,t) \in \ZZ_3 \times \ZZ_3$ to Alice and Bob respectively, and $\vec \zeta \in \Z$ (Eqn.~(\ref{eqn:zero_padded_vects})) such that $\len(\vec \zeta) = N$ to $E_1, \dots, E_N$ respectively.
        
        \item The parties have to respond with $(a, \vec \mu, b) \in \ZZ_2 \times (\ZZ_2^2)^N \times \ZZ_2$ respectively.
    \end{itemize}
\end{definition}

\begin{figure}[H]
    \centering
    \includegraphics[trim={0cm 5cm 0cm 5cm}, clip, width = \linewidth, keepaspectratio]{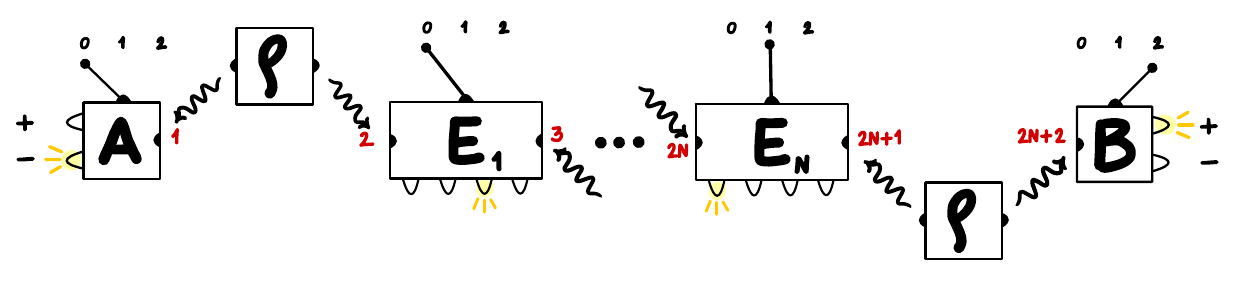}
    \caption{The $N^{\mathrm{th}}$ level of the teleportation hierarchy, $\T_N$.}
    \label{fig:Test_TN}
\end{figure}

\begin{remark}
    \begin{itemize}
        \item $\T_0$ refers to the special case where there is only one copy of $\rho$ and there are no Erwins.
        \item $\T_1$ corresponds to the group generator experiment.
        \item Given the setup for $\T_N$, every $\T_M$ for $M < N$ can be performed by the referee simply ignoring some subset of Erwins.
    \end{itemize}
\end{remark}

The fact that the correlations of the group generator experiment ($\T_1$) are preserved under iterated teleportation is encoded into the following condition, which also serves as the target correlation for each level of the hierarchy

\begin{condition}\label{cond:hierarchy_condition}
    For the teleportation hierarchy $\{\T_N\}_{N \in \NN}$
    \begin{enumerate}
        \item The dataset generated by $\T_0$ produces the correlation matrix $\Cee^{(0)}_{00}$ 
					defined in Eq.~(\ref{eqn:Corr_mat_def}).
        \item For every $N \geq 1$, the dataset generated by $\T_N$, conditioned on setting $\vec \zeta = \vec 0_{n_1} \circ r \circ \vec 0_{n_2} \in \Z,\ \len(\vec \zeta) = N$ and outcome $\vec \mu = \vec m_1 \circ \mu \circ \vec m_2$ (refer to Sec.~\ref{subsec:notation} for notation) produces the correlation matrix
        \begin{align*}
            \Cee^{(\vec \zeta)}_{\vec \mu} := \K_{\vartheta(\vec m_1)} \Cee^{(r)}_{\mu} \K_{\vartheta(\vec m_2)}.
        \end{align*}
    \end{enumerate}
\end{condition}

The rest of this section is dedicated to deriving the constraints on the (effective) GPT by virtue of it satisfying Condition~\ref{cond:hierarchy_condition}.

%========================================================
% Consequences of the hierarchy condition 
%========================================================

\subsection{Consequences of Condition~\ref{cond:hierarchy_condition}}\label{subsec:consequences_of_cond}

\begin{lemma}\label{lem:max_chsh}
    Given a choice of $\theta \in (0, \pi/4]$, the maximum CHSH value observed in $\T_N$, for any $N \in \NN$ is 
    \begin{align*}
        2\sqrt2 \cos\big( \theta - \tfrac\pi4 \big).
    \end{align*}
\end{lemma}

\begin{proof}
    The CHSH value for all possible CHSH observables for the pairs of settings $(0,1)$, $(1,2)$, and $(2,0)$ can be directly read off from the correlation matrices $\Cee^{(\vec r)}_{\vec \mu}$.
    Since the matrices $\K_\mu$ only flip the sign of some of the rows (columns), they leave the maximum CHSH value unchanged.
    In other words, the action of the $\K_\mu$ can be interpreted as a relabelling of the CHSH observable.
    Thus the maximum CHSH value depends only on the matrices $\Cee^{(r)}_{00}$.
    The maximum value both for $\Cee^{(1)}_{00}$ and $\Cee^{(2)}_{00}$ is 
    \begin{align*}
        2(\cos\theta + \sin\theta) = 2\sqrt2 \cos\big( \theta - \tfrac\pi4 \big).
    \end{align*}
\end{proof}

Thus the maximum CHSH value measurements in $\T_N$ (for any $N$) will be in the interval $(2, 2\sqrt2]$ depending on the chosen angle $\theta$.
Importantly, in every case, a CHSH inequality is violated.
This observed CHSH violation (even if arbitrarily small) is crucial since, in the case where $\theta$ is such that the CHSH value is $2$, the corresponding rotations form a finite group and do not generate $SO(3)$.

\begin{lemma}\label{lem:linear_independence}
    The vectors (corresponding to) $\{ A_\nu \}_{\nu \in \ZZ_4}$ ($A_3 := \1$) form a linearly independent set. (equivalently $\{ B_\alpha \}_{\alpha \in \ZZ_4}$)
\end{lemma}

\begin{proof}
    Follows directly from the bilinearity of $\rho$ and the matrix $\Cee^{(0)}_{00}$.
\end{proof}

The bipartite state $\rho$ defines a linear map from the effects of Alice to the states of Bob as follows
\begin{equation}\label{eqn:rho_map_def}
    \hat\rho : \P_A \to \D_B,\ e \mapsto \hat\rho(e) := \rho(e \otimes \cdot).
\end{equation}
Similarly, effects $\phi^{(r)}_\mu$ define maps from the states of Bob to the effects of Alice 
\begin{equation}\label{eqn:phi_map_def}
    \hat\phi^{(r)}_\mu : \D_B \to \P_A,\ \sigma \mapsto \hat\phi^{(r)}_\mu(\sigma) := (\sigma \otimes \cdot)(\phi^{(r)}_\mu).
\end{equation}
Using these maps, we can define the \emph{teleportation maps} of Erwin (Fig.~\ref{fig:tel_maps}) by 
\begin{equation}\label{eqn:tel_maps_def}
    \begin{alignedat}{3}
        \Lambda^{(r)}_{\mu} 
        := 
        \hat\rho\hat\phi^{(r)}_{\mu} :&& 
        \quad \D_B \quad &\to \quad \D_B,\\
        &&\rho(e \otimes \cdot) &\mapsto \swap{\rho}{\phi^{(r)}_{\mu}}{\rho}(e \otimes \cdot),
    \end{alignedat}
\end{equation}
where $\swap{\rho}{\phi}{\rho}$ is the notation for entanglement swapping established in \cite{dmello2024entanglement}.

\begin{figure}[H]
    \centering
    \includegraphics[trim={0cm 3cm 0cm 2cm}, clip, height = 1.5cm, keepaspectratio]{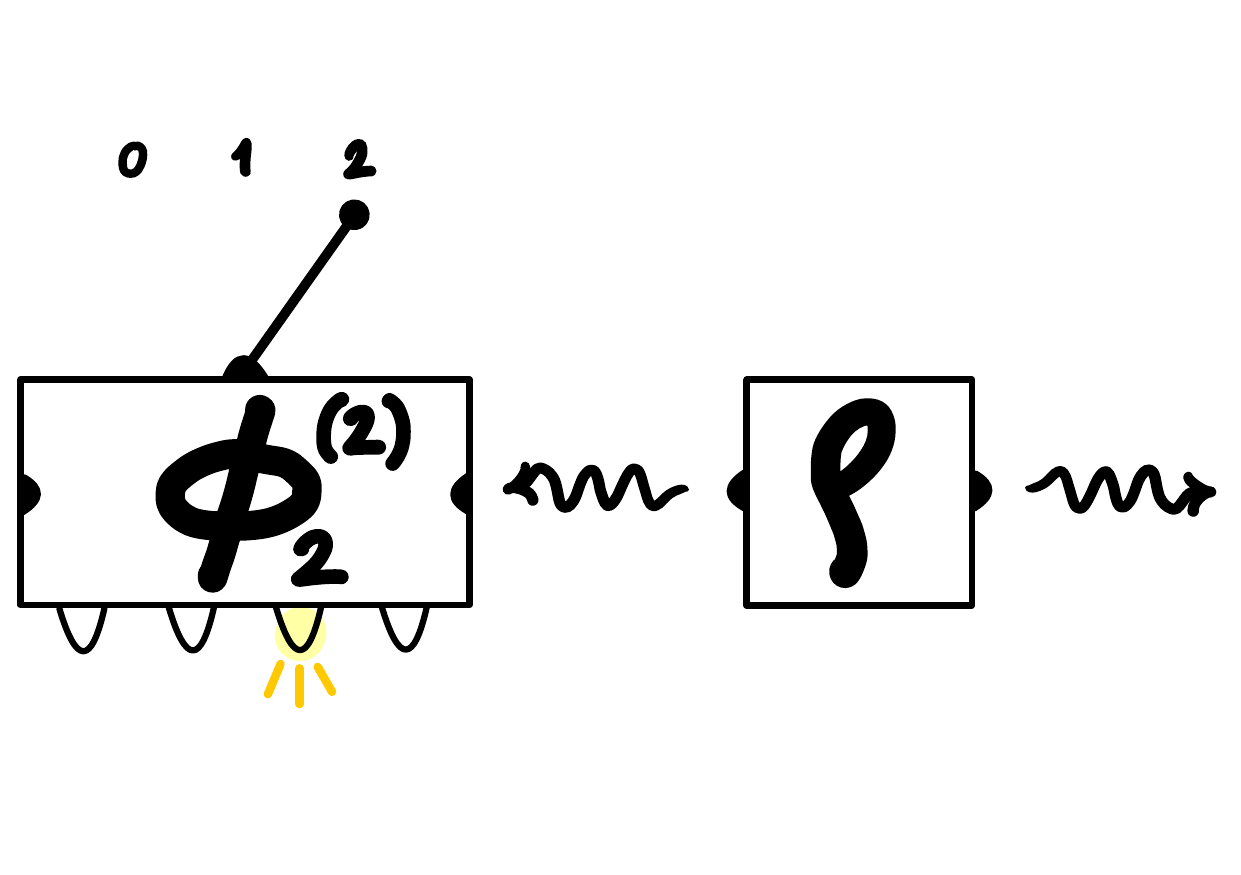}

    (a)

    \begin{align*}
        \raisebox{-0.4\height}{
            \includegraphics[trim={0cm 4cm 0cm 3cm}, clip, height = 0.75cm, keepaspectratio]{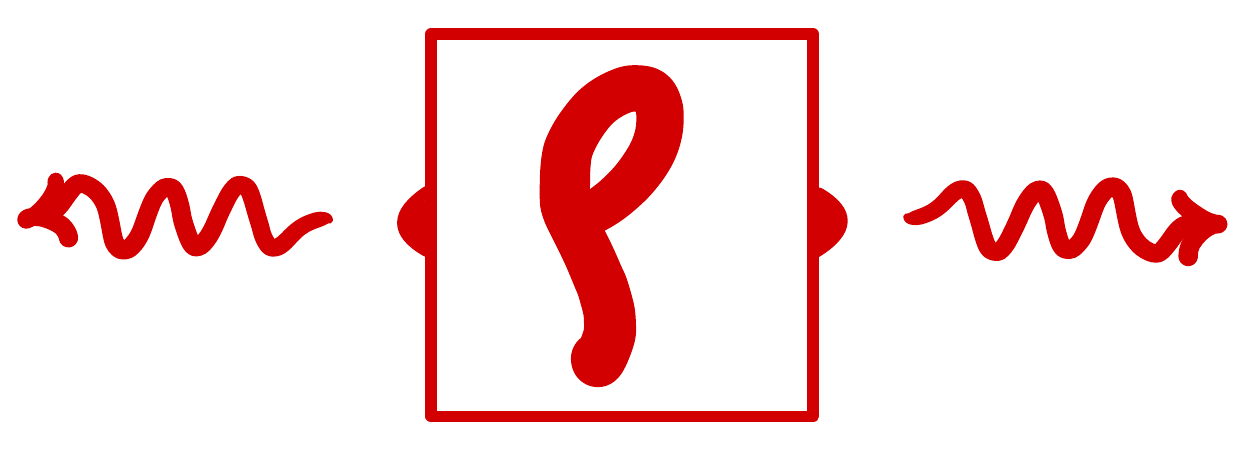}
        }
        \mapsto
        \raisebox{-0.385\height}{
            \includegraphics[trim={0cm 4cm 0cm 3.5cm}, clip, height = 1.75cm, keepaspectratio]{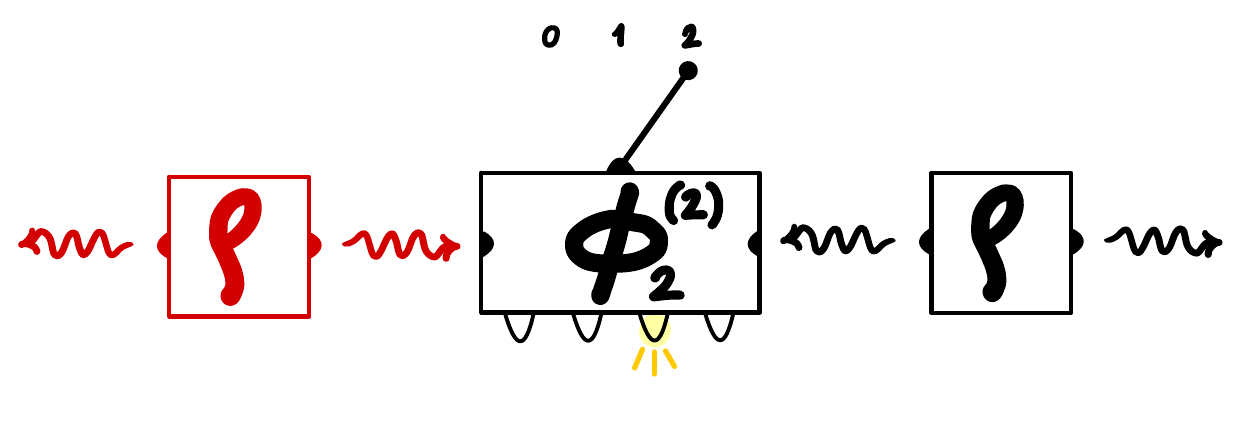}
        }
    \end{align*}
    (b)
    \caption{(a) Teleportation maps of Erwin. (b) The action of this teleportation map on bipartite states.}
    \label{fig:tel_maps}
\end{figure}

In terms of the teleportation maps, for $\T_0$ and $\T_1$, Condition~\ref{cond:hierarchy_condition} gives
\begin{gather}
    \hat\rho(A_\nu) (B_\alpha) = (\Cee^{(0)}_{00})_{\nu\alpha}, \label{eqn:T0}
    \\
    \text{and} \notag
    \\
    \Lambda^{(r)}_{\mu}\hat\rho(A_\nu) (B_\alpha) =  p(\mu | r)(\Cee^{(r)}_{\mu})_{\nu\alpha}, \label{eqn:T1}
\end{gather}
where $p(\mu|r) := \Lambda^{(r)}_{\mu}\hat\rho(\1)(\1)$ is the probability of Erwin obtaining outcome $\mu$ given setting $r$.
The definition of teleportation maps can be extended to multiple Erwins.
Given setting $\vec r$ and outcome $\vec \mu$, the resulting bipartite state between Alice and Bob is 
\begin{align*}
    \tfrac1{p(\vec \mu | \vec r)}\Lambda^{(\vec r)}_{\vec \mu}\hat\rho, \quad 
    \text{with} \quad
    \Lambda^{(\vec r)}_{\vec \mu} := \Lambda^{(r_N)}_{\mu_N} \cdots \Lambda^{(r_1)}_{\mu_1} = \prod_{i = N}^1 \Lambda^{(r_i)}_{\mu_i},
\end{align*}
and analogously $p(\vec \mu| \vec r) := \Lambda^{(\vec r)}_{\vec \mu}\hat\rho(\1)(\1)$.
In terms of these maps, Condition~\ref{cond:hierarchy_condition} implies $\forall \vec \zeta \in \Z,\ \forall \vec \mu \in (\ZZ_2^2)^{\len(\vec \zeta)},$
\begin{equation}\label{eqn:TN}
    \Lambda^{(\vec \zeta)}_{\vec \mu}\hat\rho(A_\nu)(B_\alpha) = p(\vec \mu | \vec \zeta)(\Cee^{(\vec \zeta)}_{\vec \mu})_{\nu\alpha}.
\end{equation}

\begin{lemma}\label{lem:prob_quarter_all}
    For every $\vec \zeta \in \Z$, and every $\vec \mu \in (\ZZ_2^2)^{\len(\vec \zeta)}$, 
    \begin{align*}
        p(\vec \mu | \vec \zeta) = \big(\tfrac14\big)^{\len(\vec \zeta)}.
    \end{align*}
\end{lemma}

This lemma proves three facts: Provided that Condition~\ref{cond:hierarchy_condition} is satisfied, (1) the outcomes of all the individual Erwins must be equiprobable, (2) the distribution over the joint outcomes of the Erwins must factorize, and thus (3) every (joint) outcome must occur with a non-zero probability. 

\begin{proof}
    We will prove this claim iteratively. 
    We start by proving the claim for $\len(\vec \zeta) = 1$, i.e., we show that $p(\mu|r) = \tfrac14$ for every $\mu \in \ZZ_2^2,\ r \in \ZZ_3$.

    For a given $r \in \ZZ_3$, the effects $\{\phi^{(r)}_\mu\}_{\mu \in \ZZ_4}$ describe a measurement.
    Thus they must satisfy 
    \begin{align*}
        \forall r \in \ZZ_3,\ \sum_{\mu} \phi^{(r)}_\mu = \1 \otimes \1.
    \end{align*}
    As a consequence, $\forall r \in \ZZ_3$
    \begin{align*}
        \sum_{\mu} \Lambda^{(r)}_{\mu}\hat\rho(A_\nu) (B_\alpha) 
        &
        = \rho(A_\nu \otimes \1)\rho(\1 \otimes B_\alpha)
        &
        \\
        &
        = \delta_{\nu 3} \delta_{3 \alpha}.
        &
        [\text{by Eqn.~(\ref{eqn:T0})}]
    \end{align*}
    In diagrams, $\forall r \in \ZZ_4$
    \begin{align*}
        \sum_{\mu}
        \raisebox{-0.5\height}{
            \includegraphics[trim={0cm 5.5cm 11.5cm 4.5cm}, clip, height = 1.5cm, keepaspectratio]{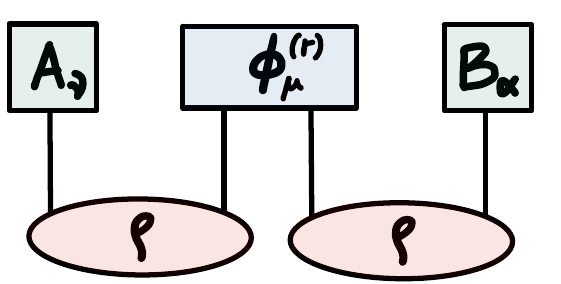}
        }
        =
        \raisebox{-0.5\height}{
            \includegraphics[trim={11.5cm 5.5cm 0cm 4.5cm}, clip, height = 1.5cm, keepaspectratio]{factorisation_1.pdf}
        }.
    \end{align*}
    Let $M$ be the matrix whose elements are $M_{\nu\alpha} = \delta_{\nu 3} \delta_{3 \alpha}$.
    Then $\forall r \in \ZZ_4$ 
    \begin{align*}
        M 
        &
        = \sum_\mu p(\mu | r) \Cee^{(r)}_\mu
        &
        [\text{by Eqn.~(\ref{eqn:T1})}]
        \\
        &
        = \sum_\mu p(\mu | r) \K_\mu \Em^{(r)}.
        &
        [\text{by Eqn.~(\ref{eqn:Corr_mat_def})}]
    \end{align*}
    Now using the fact the $\Em^{(r)}$ is invertible and $M\big( \Em^{(r)} \big)^{-1} = M$ we get 
    \begin{equation}\label{eqn:unifrom_prob_eqn_1}
        M 
        = \sum_\mu p(\mu | r) \K_\mu.
    \end{equation}
    On the other hand, since the $\K_\mu$ form a group representation, the average over the representation gives a projector onto the trivial rep, which in this case is exactly $M$, i.e., 
    \begin{equation}\label{eqn:unifrom_prob_eqn_2}
        \tfrac14 \sum_\mu \K_\mu = M.
    \end{equation}
    Combining Equations~(\ref{eqn:unifrom_prob_eqn_1})~and~(\ref{eqn:unifrom_prob_eqn_2}), the claim for $\len(\vec \zeta) = 1$ follows from the linear independence of the $\K_\mu$.

    Now consider an entanglement swapping chain of length $2$, i.e., $\len(\vec \zeta) = 2$.
    Here we will prove the claim for $\vec \zeta = (0, r)$.
    The proof for $(r, 0)$ follows analogously.
    Let the outcome vector be $\vec \mu = (\mu_1, \mu_2)$.
    Fixing $r$ and $\mu_2$ and applying the measurement condition we obtain 
    \begin{align*}
        \sum_{\mu_1} \Lambda^{(\vec \zeta)}_{\vec \mu} \hat\rho (A_\nu)(B_\alpha) 
        &
        = \rho(A_\nu \otimes \1) \Big(\Lambda^{(r)}_{\mu_2}\hat\rho(\1)(B_\alpha)\Big)
        \\
        &
        = \delta_{\nu3}\ p(\mu_2|r) \big( \Cee^{(r)}_{\mu_2} \big)_{3\alpha}
        \\
        &
        = \tfrac14 \delta_{\nu3} \delta_{3\alpha},
    \end{align*}
    where be have used Equations~(\ref{eqn:T0}),~(\ref{eqn:T1}), as well as the case $\len(\vec \zeta) = 1$.
    Diagrammatically, $\forall r \in \ZZ_3,\ \mu_2 \in \ZZ_2^2$
    \begin{align*}
        &
        \sum_{\mu_1}
        \raisebox{-0.5\height}{
            \includegraphics[trim={0cm 5.5cm 11cm 4.5cm}, clip, height = 2.25cm, keepaspectratio]{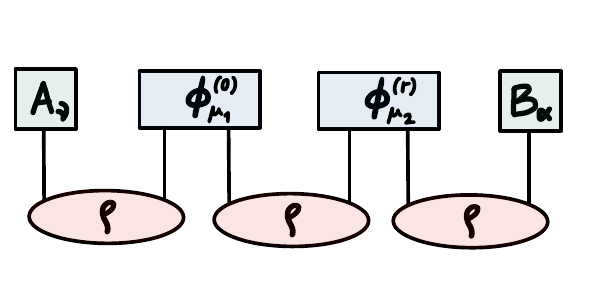}
        }
        \\
        &
        =
        \raisebox{-0.5\height}{
            \includegraphics[trim={11cm 5.5cm 0cm 4.5cm}, clip, height = 2.25cm, keepaspectratio]{factorisation_2.pdf}
        }
    \end{align*}
    Thus, by Eqn.~\ref{eqn:TN} and Condition~\ref{cond:hierarchy_condition} we get $\forall r \in \ZZ_3,\ \forall \mu_2 \in \ZZ_2^2$
    \begin{align*}
        \tfrac14 M
        &
        = \sum_{\mu_1} p(\vec \mu | \vec \zeta) \Cee^{(\vec \zeta)}_{\vec \mu}
        \\
        &
        = \sum_{\mu_1} p(\vec \mu | \vec \zeta) \K_{\mu_1}\Cee^{(r)}_{\mu_2}.
    \end{align*}
    Again, using the invertibility of $\Cee^{(r)}_{\mu_2}$ and $M(\Cee^{(r)}_{\mu_2})^{-1} = M$ we get $\forall r \in \ZZ_3,\ \forall \mu_2 \in \ZZ_2^2$
    \begin{align*}
        \sum_{\mu_1} p(\vec \mu | \vec \zeta) \K_{\mu_1} = \tfrac14 M = \big(\tfrac14\big)^2 \sum_{\mu} \K_{\mu},
    \end{align*}
    which proves the claim for $\len(\vec \zeta) = 2$.
    Now,
    \begin{enumerate}
        \item Every $\vec \zeta \in \Z$ with $\len(\vec \zeta) = N+1$ can be written as either $0 \circ \vec \upsilon$ or $\vec \upsilon \circ 0$ (here ``$\circ$'' refers to concatenation of lists) for some $\vec \upsilon \in \Z$ with $\len(\vec \upsilon) = N$.
        In other words, every $\vec \zeta \in \Z$ is obtained by padding a smaller list by $0$ either to the right or left.
        
        \item $\Cee^{(\vec \zeta)}_{\vec \mu}$ is a product of invertible block matrices (with block structure $3 \oplus 1$).
        Therefore it is always the case that $\Cee^{(\vec \zeta)}_{\vec \mu}$ is invertible and $M(\Cee^{(\vec \zeta)}_{\vec \mu})^{-1} = M$.
    \end{enumerate}
    Thus iterating the above argument by applying the measurement condition each time to the newly padded Erwin proves the claim.
\end{proof}

%========================================================
% Compact topological semigroup
%========================================================

\subsection{Compact topological semigroup}\label{subsec:comp_topo_semigrp}

Let the normalized teleportation maps be denoted by $\tilde\Lambda^{(\vec \zeta)}_{\vec \mu} := \tfrac1{p(\vec\mu | \vec \zeta)}\Lambda^{(\vec \zeta)}_{\vec \mu}$.
Define the subset $\Z_0 \subseteq \Z$ of zero-vectors
\begin{align*}
    \Z_0 := \{ \vec 0_N := (\underbrace{0, \dots, 0}_N) \ | \ N \in \NN,\ N \geq 1 \}.
\end{align*}

\begin{corollary}[to Lem.~\ref{lem:prob_quarter_all}]
    The set of maps 
    \begin{align*}
        S := \Big\{ \tilde\Lambda^{(\vec \zeta)}_{\vec \mu} \ \Big| \ \vec \zeta \in \Z_0  \Big\}
    \end{align*}
    is a semigroup, and is generated by $\{ \tilde\Lambda^{(0)}_\mu \}_{\mu \in \ZZ_2^2}$.
\end{corollary}

In this section:
\begin{itemize}
    \item We show that the closure of the above semigroup is in fact a compact topological semigroup.
    \item Then, we will then use some results about compact topological semigroups (from \cite{eisner2015operator}) to show that this semigroup contains a compact group.
    \item Finally, we will use positivity conditions arising from the fact that these mathematical objects describe physical processes to characterize this compact group. 
\end{itemize}

In order to show this we have to take one preparatory step.

\begin{lemma}\label{lem:rho_invertible}
    Without loss of generality, the map $\hat\rho$ can be assumed to be invertible.
\end{lemma}

\begin{proof}
    Let the local GPT of Alice and Bob be supported on $V_A$ and $V_B$ respectively.
    This implies that state $\rho$ is a bilinear form $\rho : V_A \times V_B \to \RR$, which gives a linear map 
    \begin{align*}
        \hat\rho : V_A \to V_B^*,\ v \mapsto \rho(v \otimes \cdot).
    \end{align*}
    Correspondingly, the effects of the Erwins are vectors in $V_B \otimes V_A$ and also define linear maps
    \begin{align*}
        \hat\phi^{(r)}_\mu : V_B^* \to V_A,\ \sigma \mapsto (\sigma \otimes \cdot) (\phi^{(r)}_\mu).
    \end{align*}
    Now, a correlation in the experiment $\T_N$ is mathematically equivalent to the following trace 
    \begin{align*}
        \tr\Big( (\hat\rho \hat E_{1}) \cdots (\hat\rho \hat E_{N+1}) \Big),
    \end{align*}
    where one of the $\hat E$ corresponds to the ``linear map'' version of $B_\alpha \otimes A_\nu$ and all the other $\hat E$ correspond to $\hat\phi^{(r)}_\mu$.
    The trace implies a loop structure of the kind:
    \begin{align*}
        \includegraphics[trim={0cm 3cm 0cm 3cm}, clip, height = 2.5cm, keepaspectratio]{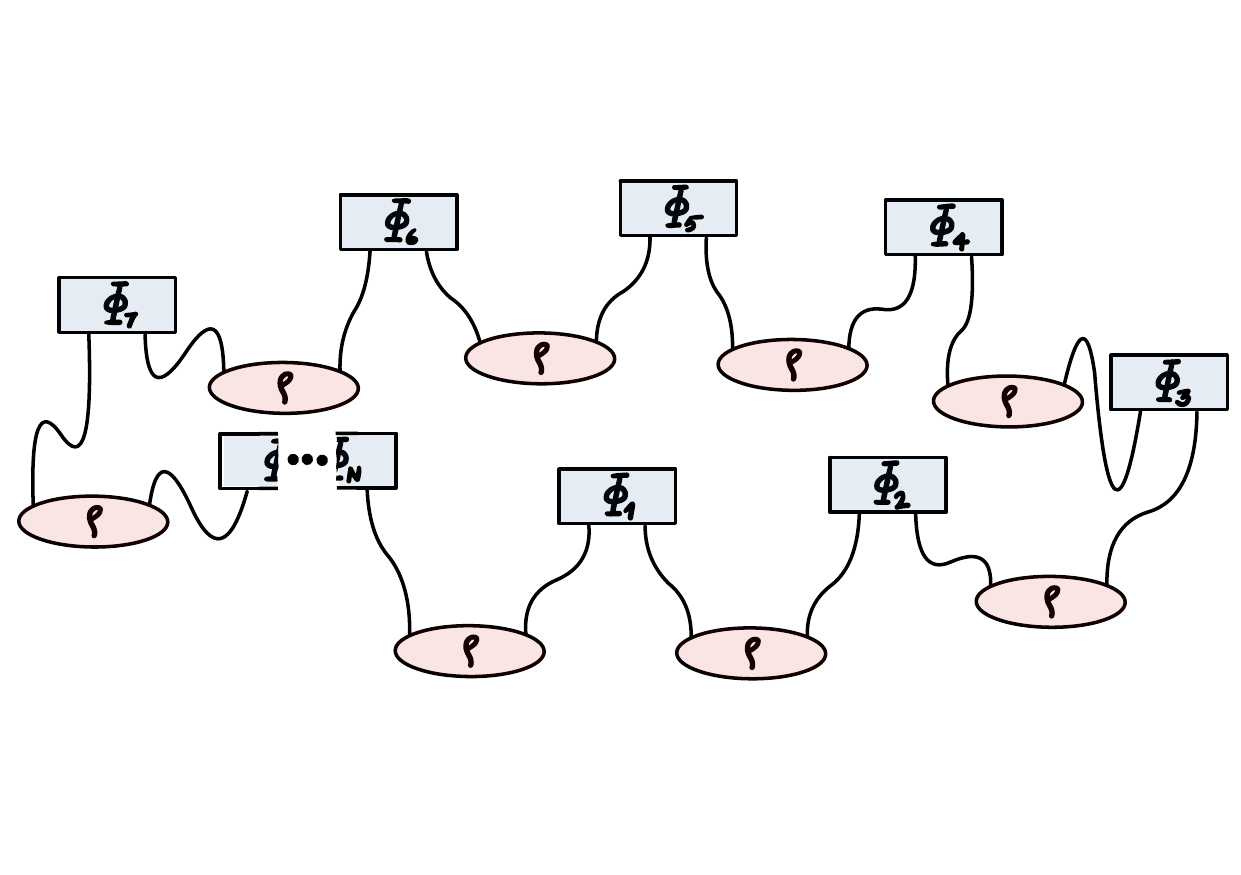}
    \end{align*}
    and the product structure of $B_\alpha \otimes A_\nu$ allows us to unravel the loop.
    
    Now, Let $\ker \hat\rho$ and $\img \hat\rho$ be the kernel and image of the map $\hat\rho$ respectively.
	Choose decompositions 
	\begin{align*}
			V_A = U \oplus \ker \hat\rho,
			\qquad
			V_B^* = \img \hat\rho \oplus W^*.
	\end{align*}
	Clearly $\hat\rho$ induces an isomorphism $\tilde\rho : U \to \img\hat\rho$.
    Let $\pi : V_A \to U$ be the projector onto $U$ along $\ker \hat\rho$, and $\iota : \img\hat\rho \to V_B^*$ be the embedding of $\img \hat\rho$ into $V_B^*$.
    We claim that restricting $V_A$ and $V_B$ to $U$ and $(\img\hat\rho)^* \subseteq V_B$ respectively, leaves all the correlation of the GPT closure of the data specifying $\T_N$ unchanged. 
    More precisely, without loss of generality, we can restrict
    \begin{gather*}
        \hat\rho \mapsto \tilde\rho, \quad
        \hat\phi^{(r)}_\mu \mapsto \pi \hat\phi^{(r)}_\mu \iota, \\
        A_\nu \mapsto A_\nu \iota, \quad 
        B_\alpha \mapsto \pi B_\alpha,
    \end{gather*}
    because
    \begin{align*}
        \tr\Big( \tilde\rho(\pi\hat E_1\iota) \cdots \tilde\rho(\pi\hat E_N\iota) \Big)
        &
        = \tr\Big( (\iota\tilde\rho\pi)\hat E_1 \cdots (\iota\tilde\rho\pi)\hat E_N \Big)
        \\
        &
        = \tr\Big( (\hat\rho\hat E_1) \cdots (\hat\rho\hat E_N) \Big),
	\end{align*}
    where we have used that $\iota \tilde\rho \pi = \hat\rho$ by construction.
\end{proof}

Thus, we assume $\hat\rho$ to be invertible.

\begin{lemma}\label{lem:comp_top_sg}
    The closure $\Ess$ of the semigroup $S$ is a compact topological semigroup.
\end{lemma}

\begin{proof}
	The elements of $\Ess$ are linear maps on a finite-dimensional vector space. 
    Hence their composition is continuous.
	Therefore $\Ess$ is a topological semigroup.
    $\Ess$ is closed by construction.
    For every $\tilde\Lambda \in S$, $\tilde\Lambda\hat\rho(\1)(\1) = 1$ by construction. 
    Therefore $S\hat\rho$ is a subset of the state space $\S^{(2)}$ which is a compact set (c.f.\ App.~\ref{app:compact_state_sapce}).
	The elements of $\Ess\hat\rho$ are limits of sequences of elements of $S\hat\rho$, and hence $\Ess\hat\rho$ is a closed subset of $\S^{(2)}$.
    Thus $\Ess\hat\rho$ is a compact set.
	Finally, $\Ess$ is simply the image of $\Ess\hat\rho$ under the right action of $\hat\rho^{-1}$, which implies $\Ess$ is compact.
\end{proof}

To proceed, we need the following result which holds for compact topological semigroups.
%(in general) and hence in particular for $\Ess$.
%The result 
It is a direct corollary of Theorem~16.3 and Lemma~16.1 of \cite{eisner2015operator}.

\begin{theorem}\label{thm:ellis}
	The semigroup $\Ess$ contains a minimal right-ideal which, in turn, contains at least one idempotent $P$.
	Moreover, $\Gee_P:=P \Ess P$ is a compact topological group, with neutral element $P$.
\end{theorem}

For the remainder of the discussion, we fix one such idempotent $P$.
The final result will not depend on which one has been chosen.

\begin{lemma}\label{lem:compact_grp_has_pos_char}
    The character $\chi_P$ afforded by $\Gee_P$ is non-negative. 
\end{lemma}

The character refers to the trace of a representation.
As indicated in the proof of Lem.~\ref{lem:rho_invertible}, the trace of a teleportation map correspond to entanglement swapping correlations, when arranged in a loop.
This imposes positivity on the trace.

\begin{proof}
    Since $\Gee_P$ is a sub-semigroup of $\Ess$, $\forall g \in \Gee_P,\ g\hat\rho \in \S^{(2)}$.
    Therefore, fixing any $\phi^{(0)}_\mu$ we have
	\begin{align*}
        \tr\big(g \hat\rho\, \hat\phi^{(0)}_\mu \big) \geq 0.
	\end{align*}
    Using the fact that $P$ is the identity of $\Gee_P$ we get 
    \begin{align*}
        \tr\big(g \hat\rho \, \hat\phi^{(0)}_\mu \big)
		&
		=  p(\mu | 0) \tr\big((P g P) \, \tilde\Lambda^{(0)}_{\mu} \big)
        \\
        &
		= p(\mu | 0) \tr\big(g \, (P \tilde\Lambda^{(0)}_{\mu} P) \big).
    \end{align*}
    $\tilde\Lambda^{(0)}_{\mu} \in \Ess$ implies $P \tilde\Lambda^{(0)}_{\mu} P \in \Gee_P$.
    Let $P \tilde\Lambda^{(0)}_{\mu} P = h \in \Gee_P$.
    Thus, for any $g'\in \Gee$, choosing $g:=g' h^{-1}$ gives
    \begin{align*}
        \chi_P(g') 
        &
        = \chi_P(g h) 
        \\
        &
        = \tr(g P \tilde\Lambda^{(0)}_{\mu} P)
        \\
        &
        = \tfrac1{p(\mu | 0)} \tr\big(g \hat\rho \, \hat\phi^{(0)}_\mu \big)
        \geq 0.
    \end{align*}
\end{proof}

Consider the map $\vartheta : (\ZZ_2^2)^* \to \ZZ_2^2$, introduced in Sec.~\ref{subsec:notation}.
This map is a (semigroup) homomorphism from the semigroup $(\ZZ_2^2)^*$ of lists with entries in $\ZZ_2^2$ under concatenation of lists, into the $\ZZ_2^2$ under bitwise addition modulo two.
This trivially defines a homomorphism from $S$ into $\ZZ_2^2$.
We show that we can extended the domain of $\vartheta$ to $\Ess$ by continuity, and then to $\Gee_P$ by restriction.

\begin{lemma}
    $\vartheta : S \to \ZZ_2^2$ can be extended to a semigroup homomorphism $\vartheta : \Ess \to \ZZ_2^2$.
    Restricting this map to the compact group $\Gee_P$ (which is a sub-semigroup of $\Ess$), obtains a group homomorphism
    \begin{align*}
        \vartheta : \Gee_P \to \ZZ_2^2.
    \end{align*}
\end{lemma}

\begin{proof}
    Consider the outcomes $\vec \mu_1, \vec \mu_2 \in (\ZZ_2^2)^*$ for settings $\vec \zeta_1, \vec \zeta_2 \in \Z_0$.
    If $\vartheta(\vec \mu_1) \neq \vartheta(\vec \mu_2)$, then there exists $\nu \in \{0,1,2\}$ such that 
    \begin{align*}
        \Big|\big(\tilde\Lambda^{(\vec \zeta_1)}_{\vec \mu_1} - \tilde\Lambda^{(\vec \zeta_2)}_{\vec \mu_2}\big)\hat\rho(A_\nu)(B_\nu)\Big| = 2.
    \end{align*}
    Thus if two elements of $S$ are sufficiently close then it must be that $\vartheta(\vec \mu_1) = \vartheta(\vec \mu_2)$.
    Therefore, by continuity every convergent sequence in $S$ will eventually map to the same element in $\ZZ_2^2$ under $\vartheta$.
    To extended $\vartheta$, simply assign the limiting value of $\vartheta(\vec \mu)$ to the limit point.
\end{proof}

\begin{remark}
    It follows that both $\Ess$ and $\Gee_P$ decompose into a set of disconnected components, on which the map $\vartheta$ is constant.    
\end{remark}

The homomorphism $\vartheta : \Gee_P \to \ZZ_2^2$ enables us to obtain a representation of $\ZZ_2^2$ by modding out the kernel of $\vartheta$ from $\Gee_P$, i.e., 
\begin{align}\label{eqn:K_4_from_compact}
    \Gee_P / \ker\vartheta \cong \ZZ_2^2.
\end{align}
Let $m$ be the Haar measure on $\ker \vartheta$.
Then 
\begin{align}\label{eqn:haar}
     \Pi_\varphi := \int_{g \in \ker \vartheta}dm(g)\,g
\end{align}
is a projection onto the trivial representation of the normal subgroup $\ker\vartheta \trianglelefteq \Gee_P$. 
Let $g \in \Gee_P$ and define the map
\begin{align*}
    \varphi : \ZZ_2^2 \to \L(\img \Pi_\varphi),\ \vartheta(g) \mapsto \Pi_\varphi g \Pi_\varphi.
\end{align*}

\begin{theorem}\label{thm:the_trace_positive_rep}
	The map $\varphi$ is a representation of the group $\ZZ_2^2$ with the following properties:
    \begin{enumerate}
		\item 
			Its character $\chi_\varphi$ is non-negative.
        \item 
			The trivial representation has multiplicity 1 in the irrep decomposition of $\varphi$.
    \end{enumerate}
\end{theorem}

Both properties relate directly to the physics of the situation.
Regarding the first point, it turns out that the values of the character can be expressed as a pairing between states and effects, and is thus non-negative.
The proof of the second fact relates the rank of the projection onto the trivial representation (of $\varphi$) to the rank of 
the marginal state shared by Alice and Bob.
The marginal state is of rank one because it factorizes by construction.

\begin{proof}
	The map $\varphi$ is the restriction of $\Gee_P$ to the subspace that is invariant under the normal subgroup $\ker \vartheta$.
	Arguing as in the proof of Lemma~\ref{lem:compact_grp_has_pos_char}, it holds that $g\rho\in \S^{(2)}$ for every $g$ that appears in the integral in Eq.~(\ref{eqn:haar}).
	Because the Haar measure is normalized and $\S^{(2)}$ a convex set, 
	$\Pi_\varphi \hat\rho$ is a state.
	In particular, this means that $\img(\Pi_\varphi)$ is non-trivial.
	Thus it is a representation of the quotient group $\ZZ_2^2$.

    For (1) we first choose a transversal of $\Gee_P/\ker \vartheta$ in order to compute the character.
    Indeed by construction, the elements of
    \begin{align*}
        \big\{ g_\mu := P\tilde\Lambda^{(0)}_{\mu}P \ | \ \mu \in \ZZ_2^2 \big\}
    \end{align*}
    are in one-one correspondence with the elements of $\ZZ_2^2$, and hence are a transversal of $\Gee_P/\ker\vartheta$.
    Now, for any $\mu \in \ZZ_2^2$ we can compute
    \begin{align*}
        \chi_\varphi(\mu) 
        = \tr(\Pi_\varphi g_\mu \Pi_\varphi) 
		&
        = \tr(g_\mu \Pi_\varphi)
        \\
		&
        = \int_{h \in \ker\vartheta}d\mu(h)\chi_P(g_\mu h) \geq 0,
    \end{align*}
    where we have used that the character $\chi_P$ is non-negative (Lem.~\ref{lem:compact_grp_has_pos_char}).

    We prove (2) by bounding the multiplicity of the trivial representation above and below by one.
    The lower bound always holds for a non-negative character of a finite group.
    This is because the character inner product of any non-negative character $\chi$ with the trivial character $\chi_1$ yields 
    \begin{align*}
        [\chi,\chi_1] 
        = \tfrac1{|G|} \sum_{g \in G} \chi_1(g)^*\chi(g)
        = \tfrac1{|G|} \sum_{g \in G} \chi(g) > 0,
    \end{align*}
    since $\chi(\Id) > 0$.

    For the upper bound, we show that (i) The trivial representation is a subspace of the $+1$ eigenspace of the map
    \begin{align*}
        \sum_{\mu \in \ZZ_2^2} p(\mu | 0) \Pi_\varphi g_\mu \Pi_\varphi,
    \end{align*}
    and that (ii) The rank of this map is upper bounded by 1.
    Indeed, by definition, for any $v$ in the trivial representation, we have $\forall \mu \in \ZZ_2^2,\ \Pi_\varphi g_\mu \Pi_\varphi \, v = v$, and thus 
    \begin{align*}
        \bigg(\sum_{\mu \in \ZZ_2^2} p(\mu|0) \Pi_\varphi g_\mu \Pi_\varphi \bigg) \, v 
        = \sum_{\mu \in \ZZ_2^2} p(\mu|0) \, v 
        = v,
    \end{align*}
    hence (i) follows.
    For (ii), using $\forall g \in \Gee_P, PgP = g$, $P^2 = P$, and $\Pi_\varphi P = P \Pi_\varphi = \Pi_\varphi$ (which is a consequence of, Eqn.~(\ref{eqn:haar}), the definition of $\Pi_\varphi$),
    we expand the map: 
    \begin{align*}
        \sum_{\mu \in \ZZ_2^2} p(\mu|0) \Pi_\varphi g_\mu \Pi_\varphi
        &
        = \sum_{\mu \in \ZZ_2^2} p(\mu|0) \Pi_\varphi P\tfrac1{p(\mu|0)}\Lambda^{(0)}_{\mu}P \Pi_\varphi
        \\
        &
        = \sum_{\mu \in \ZZ_2^2} \Pi_\varphi \Lambda^{(0)}_{\mu} \Pi_\varphi
        \\
        &
        = \Pi_\varphi \bigg(\sum_{\mu \in \ZZ_2^2}  \Lambda^{(0)}_{\mu} \bigg) \Pi_\varphi
        \\
        &
        = \Pi_\varphi \bigg(\sum_{\mu \in \ZZ_2^2} \hat\rho\hat\phi^{(0)}_{\mu} \bigg) \Pi_\varphi
        \\
        &
        = \Pi_\varphi \hat\rho \bigg(\sum_{\mu \in \ZZ_2^2} \hat\phi^{(0)}_{\mu} \bigg) \Pi_\varphi
        \\
        &
        = \Pi_\varphi \big(\rho(\1, \cdot) \1 \big) \Pi_\varphi.
    \end{align*}
    The map
    \begin{align*}
        \rho(\1, \cdot) \1 : V_B^* \to V_B^*,\  \sigma \mapsto \sigma(\1) \rho(\1, \cdot)
    \end{align*}
    is manifestly of rank one, and hence $\Pi_\varphi \big(\rho(\1, \cdot) \1 \big) \Pi_\varphi$ is of rank at most one.
    Thus (ii) and (2) follow.
\end{proof}

Crucially, the projector $\Pi_\varphi$ onto the representation space of $\varphi$ is obtained as a limit of long chains of teleportations.
Referring back to Condition~\ref{cond:hierarchy_condition}, since the correlations are preserved under iterated teleportation, by continuity it must be the case that they are also preserved when restricting to the representation space of $\varphi$.
This almost gives us our effective GPT, but for the fact that the constraints derived so far is on the states of Bob.

In Sec.~\ref{subsec:consequences_of_cond} (Eqn.~(\ref{eqn:tel_maps_def})) we defined the teleportation maps to be 
\begin{align*}
    \hat\rho\hat\phi^{(r)}_{\mu} : V_B^* \to V_B^*.
\end{align*}
But we could have equivalently defined them as (Fig.~\ref{fig:dual_tel_maps})
\begin{align*}
    \hat\phi^{(r)}_{\mu}\hat\rho : V_A \to V_A.
\end{align*}

\begin{figure}[H]
    \centering
    \includegraphics[trim={0cm 3cm 0cm 2cm}, clip, height = 1.5cm, keepaspectratio]{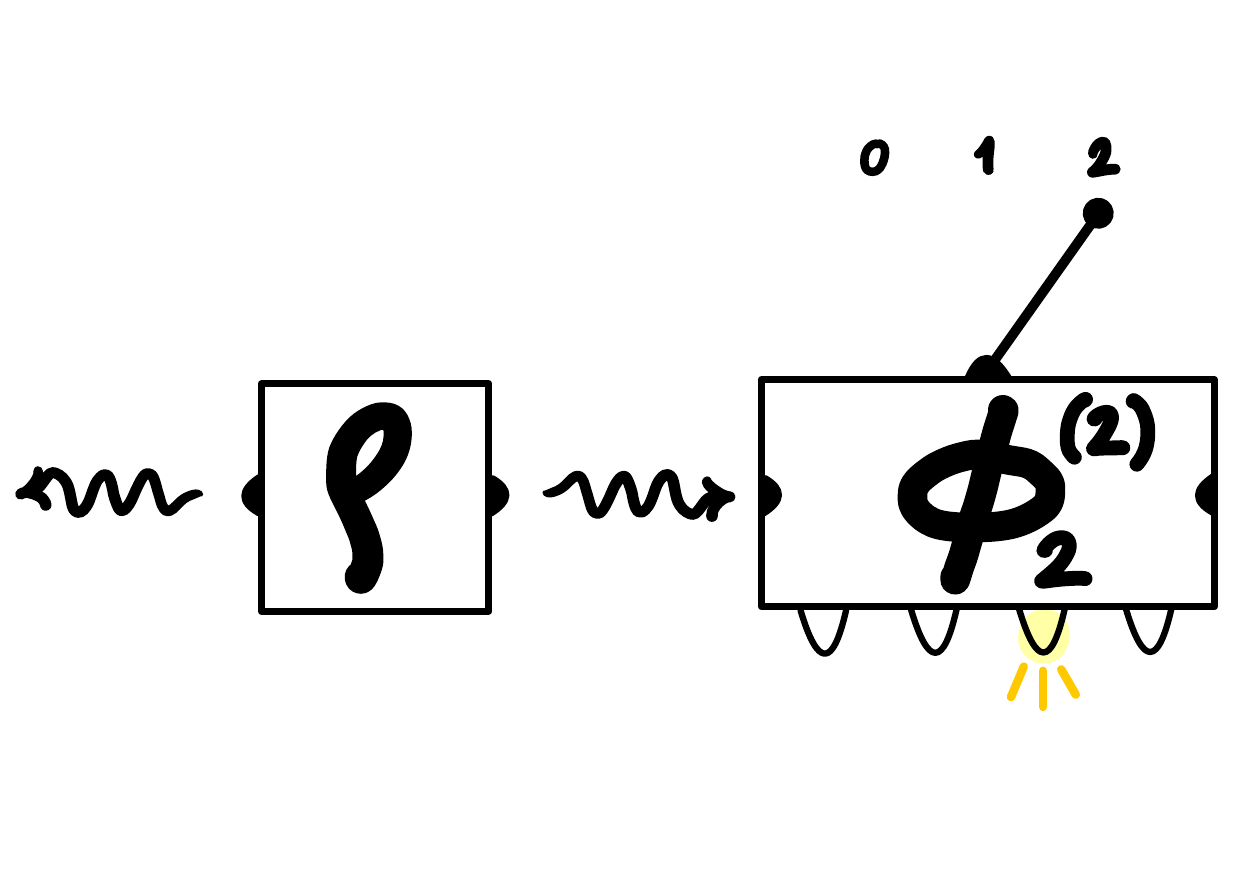}

    (a)

    \begin{align*}
        \raisebox{-0.25\height}{
            \includegraphics[trim={2cm 0cm 3cm 0cm}, clip, height = 1.35cm, keepaspectratio]{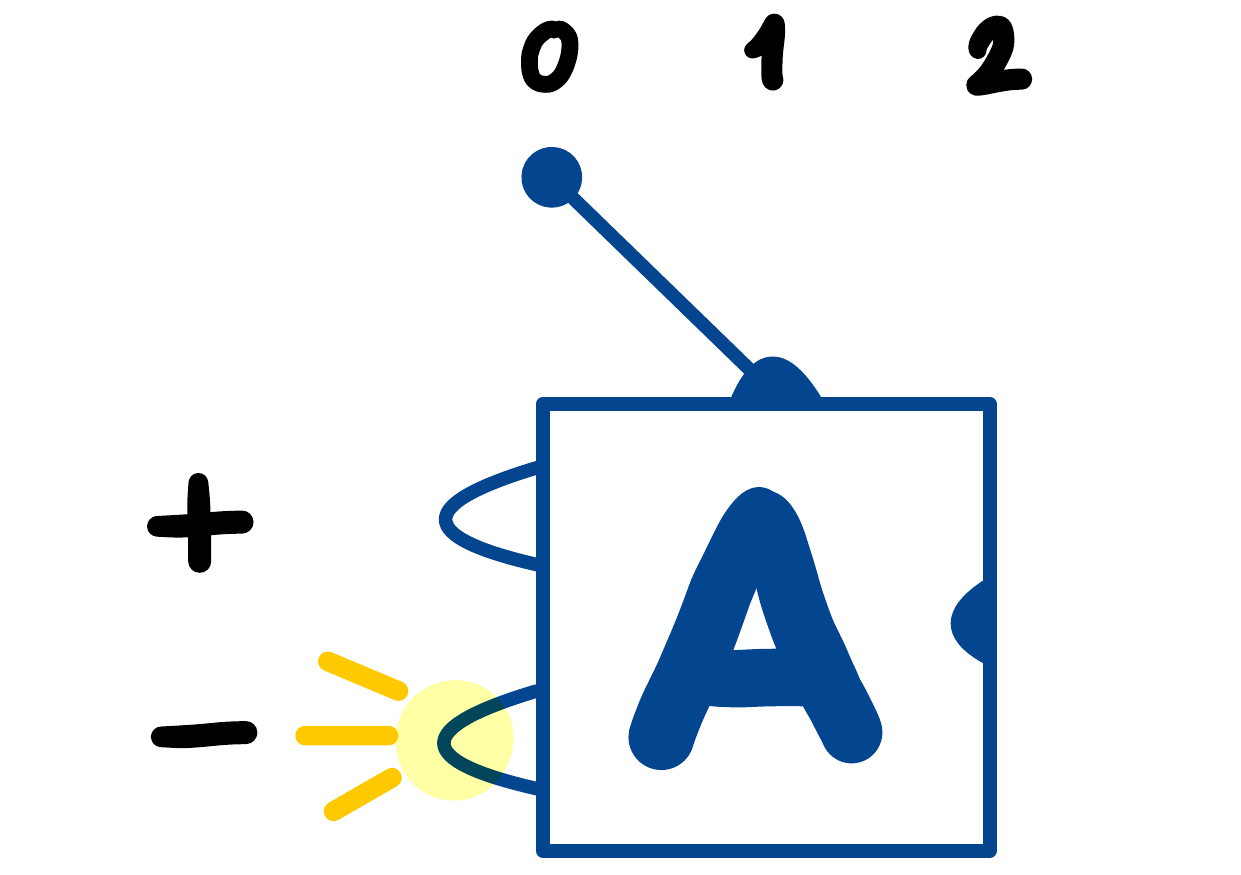}
        }
        \mapsto
        \raisebox{-0.37\height}{
            \includegraphics[trim={0cm 4cm 0cm 3.5cm}, clip, height = 1.75cm, keepaspectratio]{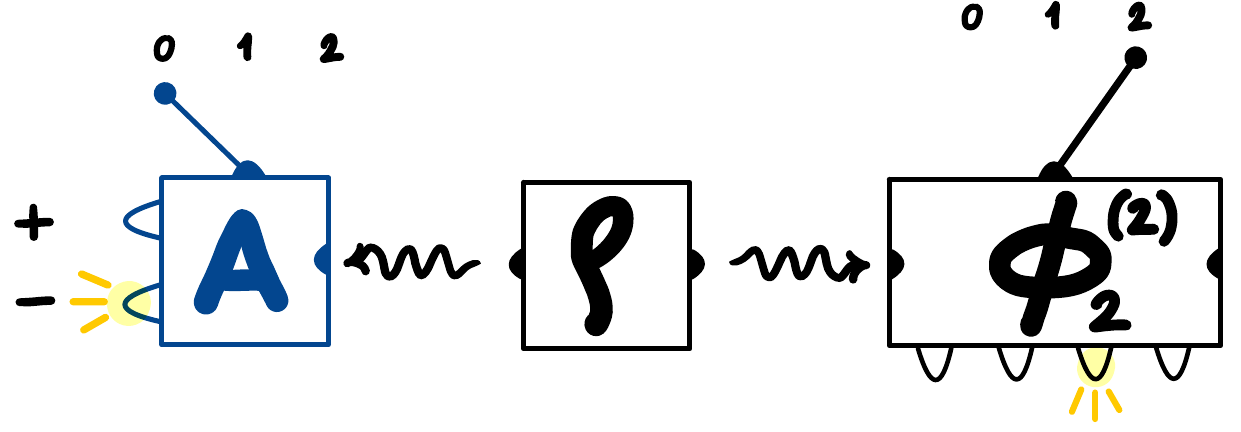}
        }
    \end{align*}
    (b)
    \caption{(a) Equivalent definition of the teleportation maps of Erwin. (b) The action of these teleportation map on the effects of Alice.}
    \label{fig:dual_tel_maps}
\end{figure}

That is, instead of defining the teleportation maps as teleporting $B$-type states, we could have defined them as teleporting $A$-type measurement effects.
Although, in general, this could make a difference, due to the invertibility of $\hat\rho$, the two is equivalent, i.e., 
\begin{align*}
    \hat\phi^{(r)}_{\mu}\hat\rho = \hat\rho^{-1} \ (\hat\rho\hat\phi^{(r)}_{\mu}) \ \hat\rho.
\end{align*}
Thus we can define a representation, equivalent to $\varphi$, on the observables of Alice:
\begin{align*}
    \forall \mu \in \ZZ_2^2,\ \pi(\mu) := \hat\rho^{-1} \varphi(\mu) \hat\rho,
\end{align*}
and the corresponding projector $\Pi_\pi := \hat\rho^{-1} \Pi_\varphi \hat\rho$.

This enables us to define the \emph{effective GPT}
\begin{theorem}[The effective GPT]\label{thm:effective_gpt}
    If the observables and state $A_\nu, B_\alpha, \phi^{(r)}_\mu, \rho$, from any GPT, satisfy Condition~\ref{cond:hierarchy_condition}, then the restrictions 
    \begin{gather*}
        A_\nu \mapsto \Pi_\pi A_\nu,
        \qquad
        B_\alpha \mapsto \Pi_\varphi^t B_\alpha,
        \\
        \hat\rho \mapsto \Pi_\varphi \hat\rho \Pi_\pi,
        \quad
        \text{and}
        \quad
        \hat\phi^{(r)}_\mu \mapsto \Pi_\pi\hat\phi^{(r)}_\mu\Pi_\varphi,
    \end{gather*}
    also satisfy Condition~\ref{cond:hierarchy_condition}.
    We will refer to the GPT closure of the restricted objects above as \emph{the effective GPT}.
\end{theorem}

To conclude the section we deduce what $\varphi$ should be.
The representation $\varphi$ of $\ZZ_2^2$ is subject to the following constraints:
\begin{enumerate}
    \item $\varphi$ has dimension at least four in order to support the correlation matrices.
	\item The character of the representation $\chi_\varphi$ has to be non-negative (Thm.~\ref{thm:the_trace_positive_rep}).
	\item The trivial character $\chi_1$ occurs exactly once in the decomposition of $\chi_\varphi$ (Thm.~\ref{thm:the_trace_positive_rep}).
\end{enumerate}
This is sufficient to conclude that $\varphi$ must correspond to the regular representation of $\ZZ_2^2$.
In fact, this is a special case of the lemma that follows.
We prove the general case as it will be useful in the later part of this paper.

\begin{lemma}\label{lem:reg_rep_proof}
    Let $\varphi$ be a representation of a finite group $G$, and $\chi$ its character. 
    Then $\varphi$ is the regular representation of $G$ iff 
    \begin{enumerate}
        \item The dimension of the representation is greater than or equal to $|G|$,
        \item The trivial rep has multiplicity $1$ in the irrep decomposition of $\varphi$,
        \item The character $\chi$ is non-negative.
    \end{enumerate}
\end{lemma}

\begin{proof}
    Let $\varphi$ be the regular representation of a finite group $G$, and $\chi$ its character. 
    The regular representation (a) is defined by the left action of the group on the group algebra $\CC[G]$, and 
    (b) contains each irrep with the multiplicity equal to the dimension of the irrep.
    (a) implies that the dimension of $\dim(\varphi) = |G|$, since $\dim(\CC[G]) = |G|$.
    (b) implies that the trivial representation had multiplicity $1$ in the irrep decomposition of $\varphi$ since the trivial representation is one-dimensional.
    Finally, the regular character of any group is 
    \begin{align*}
        \chi_{\mathrm{reg}}(g) = |G|\delta_{g, \Id},
    \end{align*}
    and is therefore non-negative.

    Conversely, assume that (i) $\dim(\varphi) \geq |G|$, (ii) the multiplicity of the trivial rep in $\varphi$ is $1$, and (iii) that the character $\chi$ of $\varphi$ is non-negative.
    Taking the character inner product with the trivial character, (ii) implies
    \begin{align*}
        1 = [\chi_1, \chi]
        &
        = \tfrac1{|G|} \sum_{g \in G} \overline{\chi_1(g)} \chi(g)
        = \tfrac1{|G|} \sum_{g \in G} \chi(g),
    \end{align*}
    which implies 
    \begin{align*}
        \sum_{g \in G} \chi(g) = |G|.
    \end{align*}
    Now, from (i), we have the relation 
    \begin{align*}
        \chi(\Id) =  \dim(\varphi) \geq |G|.
    \end{align*}
    Combining (i) and (ii) therefore yields 
    \begin{align*}
        |G| 
        % &
        = \sum_{g \in G} \chi(g)
        % \\
        &
        =  \chi(\Id) + \sum_{g \in G, g \neq \Id} \chi(g)
        \\
        &
        \geq  |G| + \sum_{g \in G, g \neq \Id} \chi(g),
    \end{align*}
    which implies 
    \begin{align*}
        \sum_{g \in G, g \neq \Id} \chi(g) \leq 0.
    \end{align*}
    Finally by (iii), $\chi(g) \geq 0$ which implies 
    \begin{align*}
        \forall g \neq \Id,\ \chi(g) = 0.
    \end{align*}
    This also gives $\chi(\Id) = |G|$, which implies $\chi$ is the regular character.
\end{proof}

%================================================================================================
% OBTAINING THE QUBIT
%================================================================================================

\section{Obtaining the qubit}\label{sec:obtaining_the_qubit}

The goal of this section is to derive an isomorphism between the effective GPT described in Thm.~\ref{thm:effective_gpt} and qubit quantum theory.
The main strategy will be to show that, in the effective GPT, the correlation between Alice and Bob after teleportation through multiple Erwins 
%(with setting $r \in \ZZ_3$ unrestricted), 
can be obtained by multiplying correlation matrices.
Since the generators of the qubit symmetry were encoded in the correlations, we obtain the full symmetry group of the qubit state space as teleportation maps.
Finally, because the orbit under this group corresponds to physical operation on the states (equivalently, observables) of the effective GPT, the claim is proven.

To start, we introduce notation to talk about the states and effects in the effective GPT.
Let
\begin{align*}
    \hat\varrho := \Pi_\varphi \hat\rho \Pi_\pi, \quad 
    \hat\Phi^{(r)}_\mu := \Pi_\pi\hat\phi^{(r)}_\mu\Pi_\varphi.
\end{align*}
We will denote $A_\nu \equiv \Pi_\pi A_\nu$ and $B_\alpha \equiv \Pi_\varphi^t B_\alpha$.
Recall that Lem.~\ref{lem:prob_quarter_all} (equiprobability of the outcomes of Erwins) relied only on the structure of correlations in Condition~\ref{cond:hierarchy_condition} and the measurement condition ($\sum_\mu \phi^{(r)}_\mu = \1 \otimes \1$).
Thus Lem.~\ref{lem:prob_quarter_all} still holds in the effective GPT, implying
\begin{align*}
    (\hat\varrho\hat\Phi^{(r)}_\mu)\ \hat\varrho (\1)(\1) = p(\mu|r) = \tfrac14. 
\end{align*}
Therefore we directly introduce the \emph{normalized effective teleportation maps}, denoted by
\begin{align*}
    \Oh^{(r)}_\mu := \tfrac1{p(\mu|r)}\hat\varrho\hat\Phi^{(r)}_\mu.
\end{align*}
$\Oh^{(r)}_\mu$ is in fact the asymptotic version of $\tilde\Lambda^{(r)}_\mu$:
\begin{align*}
    \Oh^{(r)}_\mu
    &
    = \tfrac1{p(\mu|r)}\hat\varrho\hat\Phi^{(r)}_\mu
    \\
    &
    = \tfrac1{p(\mu|r)}(\Pi_\varphi\hat\rho \Pi_\pi)(\Pi_\pi \hat\phi^{(r)}_\mu \Pi_\varphi)
    \\
    &
    = \tfrac1{p(\mu|r)}\Pi_\varphi\hat\rho \Pi_\pi \hat\phi^{(r)}_\mu \Pi_\varphi
    \\
    &
    = \tfrac1{p(\mu|r)}\Pi_\varphi\hat\rho (\hat\rho^{-1} \Pi_\varphi \hat\rho) \hat\phi^{(r)}_\mu \Pi_\varphi
    \\
    &
    = \tfrac1{p(\mu|r)}\Pi_\varphi  \Pi_\varphi \Lambda^{(r)}_\mu \Pi_\varphi
    \\
    &
    = \Pi_\varphi \tilde \Lambda^{(r)}_\mu \Pi_\varphi.
\end{align*}
Thus, in the special case $r = 0$ we have:
\begin{equation}\label{eqn:tel_to_K_4}
    \Oh^{(0)}_\mu = \Pi_\varphi \tilde \Lambda^{(0)}_\mu \Pi_\varphi = \varphi(\mu).
\end{equation}
On the other hand, by continuity, on the effective GPT Condition~\ref{cond:hierarchy_condition} gives 
\begin{gather}
    \hat\varrho(A_\nu)(B_\alpha) = (\Cee^{(0)}_{00})_{\nu\alpha},\label{eqn:asymp_T0}
    \\
    \text{and} \notag
    \\
    \Oh^{(r)}_\mu\hat\varrho(A_\nu)(B_\alpha) = (\Cee^{(r)}_{\mu})_{\nu\alpha}.\label{eqn:asymp_T1}
\end{gather}
Eqn.~(\ref{eqn:asymp_T0}) implies that $\{ A_\nu \}_{\nu \in \ZZ_4}$ and $\{ B_\alpha \}_{\alpha \in \ZZ_4}$ are linearly independent.
Since $\dim \varphi = 4$, we can conclude that the $A_\nu$ and $B_\alpha$ span the representation space of $\pi$ and (the dual of) $\varphi$ respectively.
Further, Eqn.~(\ref{eqn:tel_to_K_4}) together with Eqns.~(\ref{eqn:asymp_T0})~and~(\ref{eqn:asymp_T1}) allows us to read off the action of $\pi(\mu)$ on the $A_\nu$ directly from the matrix $\Cee^{(0)}_\mu = \K_\mu\Em^{(0)}$:
\begin{equation}\label{eqn:rep_action_on_A}
    \forall \mu \in \ZZ_2^2,\ \pi(\mu) A_\nu = \sum_{\beta}(\K_{\mu})_{\nu\beta} A_\beta.
\end{equation}

In summary, so far we have managed to relate the bilinear forms (describing bipartite states) to correlation matrices.
But, what we want instead is to be able to relate the teleportation maps to correlations.
For this, we have extract the contribution of $\hat\varrho$ to the correlation matrices and then nullify it.

To this end, denote $W_\pi = \Span\{A_\nu\}_{\nu \in \ZZ_4} = \img(\Pi_\pi)$.
Let $\langle \cdot, \cdot \rangle_\pi$ be the inner product on $W_\pi$ with respect to which the representation $\pi$ unitary.
This inner product is unique and is obtained by twirling any inner product on $W_\pi$, i.e., let $\langle \cdot, \cdot\rangle$ be any inner product on $W_\pi$, then 
\begin{align*}
    \langle \xi, \eta\rangle_\pi := \tfrac14\sum_{g \in \ZZ_2^2} \langle \pi(g) \xi, \pi(g) \eta \rangle.
\end{align*}
Since, $\K_\mu$ is diagonal with $\pm1$ entries, it is straightforward to verify that the $A_\nu$ form an orthonormal basis with respect to $\langle \cdot, \cdot \rangle_\pi$.
This gives us two things, (i) $W_\pi \cong W_\pi^*$ by Fr\'echet-Riesz, and (ii) the ONB $A_\nu$ can be thus identified with its dual basis.
But, the observable $B_\alpha$ define linear functionals $\varrho(\cdot \otimes B_\alpha)$ on $W_\pi$ via $\varrho$.
Since $\varrho$ is non-degenerate with respect to the $A_\nu$ and $B_\alpha$, we can identify $\varrho$ with an invertible linear map $M : W_\pi \to W_\pi$ defined by the following equation:
\begin{align*}
    \langle A_\nu, M(A_\alpha) \rangle_\pi := \varrho(A_\nu \otimes B_\alpha).
\end{align*}
By Eqn.~(\ref{eqn:asymp_T0}),
\begin{align*}
    \langle A_\nu, M(A_\alpha) \rangle_\pi := \Big( \Cee^{(0)}_{00} \Big)_{\nu\alpha}.
\end{align*}
In other words, the matrix elements of $M$ in the basis $A_\nu$ are given by $\Cee^{(0)}_{00}$.
We will use the notation $\hat\Cee^{(0)}_{00} = M$.

We are now in a position to associate linear maps $W_\pi^* \to W_\pi^*$ with the normalized teleportation maps $\Oh^{(r)}_\mu$.
In order to do this, first define $\hat\Cee^{(r)}_\mu : W_\pi \to W_\pi$ via Eqn.~(\ref{eqn:asymp_T1}) by 
\begin{align*}
    \langle A_\nu, \hat\Cee^{(r)}_\mu(A_\alpha) \rangle_\pi := \Oh^{(r)}_\mu\hat\varrho(A_\nu)(B_\alpha) = \Big( \Cee^{(r)}_\mu \Big)_{\nu\alpha}
\end{align*}
Finally, define $\hat O^{(r)}_\mu := \hat\Cee^{(r)}_\mu\hat\Cee^{(0)}_{00}$, and let $O^{(r)}_\mu$ be the corresponding matrix.
Indeed, the linear map associated with $\Oh^{(r)}_\mu$ is none other than the transpose of $\hat O^{(r)}_\mu$ since 
\begin{align*}
    \Oh^{(r)}_\mu\hat\varrho(A_\nu)(B_\alpha) 
    &
    = \langle A_\nu, \hat O^{(r)}_\mu \hat\Cee^{(0)}_{00}(A_\alpha) \rangle_\pi
    \\
    &
    = \langle (\hat O^{(r)}_\mu)^T A_\nu,  \hat\Cee^{(0)}_{00}(A_\alpha) \rangle_\pi
    \\
    &
    = \hat\varrho\Big((\hat O^{(r)}_\mu)^T A_\nu\Big)(B_\alpha).
\end{align*}
Now, using that fact that $\varrho(\cdot \otimes B_\alpha)$ span $W_\pi^*$ (non-degeneracy of $\varrho$ together with linear independence of $B_\alpha$) we can go further and conclude that 
\begin{equation}\label{eqn:action_of_teleportation_matrix}
    \Oh^{(r)}_\mu\hat\varrho(A_\nu) 
    = \hat\varrho\Big(\big(\hat O^{(r)}_\mu \big)^T A_\nu\Big).
\end{equation}
This enables us to finally prove the claim at the start of this section 

\begin{theorem}[Multiplicative correlations]\label{thm:multiplicative_correlation}
    Define the maps 
    \begin{align*}
        \Oh^{(\vec r)}_{\vec \mu} := \prod_{i = \mathrm{len}(\vec r)}^{1}\Oh^{(r_i)}_{\mu_i},\quad
        \text{and} \quad 
        \hat O^{(\vec r)}_{\vec \mu} := \prod_{i = 1}^{\mathrm{len}(\vec r)}\hat O^{(r_i)}_{\mu_i}.
    \end{align*}
    If Condition~\ref{cond:hierarchy_condition} is satisfied, then $\forall N \in \NN,\ \forall \vec r \in \ZZ_3^N,\ \forall \vec \mu \in (\ZZ_2^2)^N,$
    \begin{align*}
        \Oh^{(\vec r)}_{\vec \mu}\hat\varrho(A_\nu)(B_\alpha) = \Big( O^{(\vec r)}_{\vec \mu}\Cee^{(0)}_{00} \Big)_{\nu\alpha}.
    \end{align*}
\end{theorem}

To reemphasize, Thm.~\ref{thm:multiplicative_correlation} enables us to make statements about correlations we have never tested for in hierarchy of Def.~\ref{def:teleportation_hierarchy}.
It is this property that enables us to establish the isomorphism to qubit quantum theory without having to \emph{explicitly test for every possible preparation procedure}.

\begin{proof}
    First calculate (using dirac notation to avoid confusion with indices)
    \begin{align*}
        \big(\hat O^{(r)}_{\mu}\big)^T |A_\nu\rangle
        &
        =
        \sum_{\beta} |A_\beta\rangle\langle A_\beta| \big(\hat O^{(r)}_{\mu}\big)^T |A_\nu\rangle
        \\
        &
        =
        \sum_{\beta} |A_\beta\rangle  \big(O^{(r)}_{\mu}\big)^T_{\beta\nu}
        \\
        &
        =
        \sum_{\beta} |A_\beta\rangle  \big(O^{(r)}_{\mu}\big)_{\nu\beta}
    \end{align*}
    Thus Eqn.~(\ref{eqn:action_of_teleportation_matrix}) gives 
    \begin{align*}
        \Oh^{(r)}_\mu\hat\varrho(A_\nu) 
        = \sum_{\beta} \big(O^{(r)}_{\mu}\big)_{\nu\beta} \hat\varrho(A_\beta).
    \end{align*}
    Applying $\Oh^{(r')}_{\mu'}$ to both sides of this equation gives 
    \begin{align*}
        \Oh^{(r')}_{\mu'}\Oh^{(r)}_\mu\hat\varrho(A_\nu) 
        &
        = \sum_{\beta} \Oh^{(r')}_{\mu'} \hat\varrho(A_\beta)\big(O^{(r)}_{\mu}\big)_{\nu\beta}
        \\
        &
        = \sum_{\beta, \gamma} \big(O^{(r')}_{\mu'}\big)_{\beta\gamma} \hat\varrho(A_\gamma)\big(O^{(r)}_{\mu}\big)_{\nu\beta}.
    \end{align*}
    Evaluating both sides against $B_\alpha$ gives 
    \begin{align*}
        \Oh^{(r')}_{\mu'}\Oh^{(r)}_\mu\hat\varrho(A_\nu)(B_\alpha)
        &
        = \sum_{\beta, \gamma} \hat\varrho(A_\gamma)(B_\alpha)  \big(O^{(r)}_{\mu}\big)_{\nu\beta} \big(O^{(r')}_{\mu'}\big)_{\beta\gamma}
        \\
        &
        = \sum_{\beta, \gamma} \big(\Cee^{(0)}_{00}\big)_{\gamma\alpha}  \big(O^{(r)}_{\mu}\big)_{\nu\beta} \big(O^{(r')}_{\mu'}\big)_{\beta\gamma}
        \\
        &
        = \big( O^{(r)}_{\mu} O^{(r')}_{\mu'} \Cee^{(0)}_{00}\big)_{\nu\alpha}.
    \end{align*}
    Iterating this proves the claim.
\end{proof}

\begin{lemma}\label{lem:SO(3)}
    The linear maps $\hat O^{(1)}_{00}$ and $\hat O^{(2)}_{00}$ generate a dense subset of a representation of $SO(3)$, with respect to the inner product $\langle\cdot,\cdot\rangle_\pi$, with block structure 
    \begin{align*}
        \begin{pmatrix}
            R & 0 \\
            0 & 1
        \end{pmatrix},
    \end{align*}
    where $R$ is a $3 \times 3$ block.
    That is, one copy of the usual three-dimensional irrep on the block $R$, and one copy of the trivial rep.
\end{lemma}

In fact, all matrices $O^{(r)}_{\mu}$ belong to $SO(3)$.

\begin{proof}
    By definition, the matrix elements of $\hat O^{(1)}_{00}$ in the basis $\{A_\nu\}_{\nu \in \ZZ_4}$ are given by $\Big( \Cee^{(1)}_{00}\Cee^{(0)}_{00} \Big)_{\nu\alpha}$.
    This gives (Sec.~\ref{subsec:quantum_realization_for_qubits})
    \begin{align*}
        O^{(1)}_{00} 
        =
        \begin{pmatrix}
            R^z_\theta & 0 \\
            0 & 1
        \end{pmatrix}.
    \end{align*}
    Similarly 
    \begin{align*}
        O^{(2)}_{00} 
        =
        \begin{pmatrix}
            R^x_\theta & 0 \\
            0 & 1
        \end{pmatrix}.
    \end{align*}
    Again, by definition (Sec.~\ref{subsec:quantum_realization_for_qubits}), the group generated by $R^z_\theta$ and $R^x_\theta$ is dense in the standard irrep of $SO(3)$, which proves the claim.
\end{proof}

Thus we have shown that the entire symmetry group of the qubit state space can be obtained by teleportation, this allows us to conclude the following

\begin{theorem}[Noslerist bound]\label{thm:noslerist_bound}
    The CHSH value of any GPT which satisfies Condition~\ref{cond:hierarchy_condition} is \emph{bounded below} by $2\sqrt2$.
    This bound is achievable.
\end{theorem}

There are two things of note here. 
The first point is that the above statement holds not only for the effective GPT, but also for the ``parent'' GPT.
The second point is that Thm.~\ref{thm:noslerist_bound} tells us that we can obtain a certificate for a CHSH value of $2\sqrt2$ (which is  the maximal quantum value) without ever having measured it.
In particular, this would hold even if the CHSH value we measured in the teleportation hierarchy (Def.~\ref{def:teleportation_hierarchy}, Lem.~\ref{lem:max_chsh}) is $2 + \varepsilon$ for $\varepsilon > 0$ arbitrarily small.
The only requirement is that a CHSH inequality is violated.

\begin{proof}
    By Lem.~\ref{lem:SO(3)}, the linear map $\hat O_{\pi/4}$ corresponding to the matrix
    \begin{align*}
        O_{\pi/4}
        =
        \begin{pmatrix}
            R^z_{\pi/4} & 0 \\
            0 & 1
        \end{pmatrix}
    \end{align*}
    is in the closure of the group generated by $\hat O^{(1)}_{00}$ and $\hat O^{(2)}_{00}$.
    Therefore, by the compactness of the state space, there exists a state $\varrho_{\pi/4}$ in the effective AGGs GPT such that 
    \begin{align*}
        \varrho_{\pi/4} (A_\nu \otimes B_\alpha) = \langle A_\nu, \hat O_{\pi/4} \hat\Cee^{(0)}_{00}(A_\alpha) \rangle_\pi
    \end{align*}
    which gives CHSH value $2\sqrt2$ with respect to the CHSH observable
    \begin{align*}
        A_0 \otimes B_0 
        - A_0 \otimes B_1 
        - A_1 \otimes B_0 
        - A_1 \otimes B_1.
    \end{align*}
\end{proof}

\begin{theorem}\label{thm:isomorphic_to_qubit}
    Given any instance of the group generator experiment (for the qubit), such that Condition~\ref{cond:hierarchy_condition} is satisfied, the effective GPT thus obtained (Thm.~\ref{thm:effective_gpt}) has the following properties:
    \begin{enumerate}
        \item The local state and space (of both Alice and Bob) is isomorphic to state and effect space of one qubit.
        \item The bipartite state space is isomorphic to the convex hull of the two-qubit product states and the two-qubit maximally entangled states. 
        Equivalently, the bipartite effect space contains only (and all) product and maximally entangled measurements.
    \end{enumerate}
\end{theorem}

\begin{proof}
    Identifying $A_0 \equiv \sigma_x$, $A_1 \equiv \sigma_y$, $A_2 \equiv \sigma_z$, and $\1 \equiv \Id$ gives a bijection between the (affinely shifted) Bloch ball and the local state space of the effective GPT due to the fact that the representation of $SO(3)$ on both are identical.
    With respect to this identification, the bipartite state $\varrho$ is isomorphic (as a bilinear form) to the Bell state.
    This, together with the fact that the $SO(3)$ representation corresponds to a local unitary action on the quantum side, we can generate any state with the same singular values (Schmidt coefficients) as the Bell state, i.e., all the maximally entangled states.
    An analogous arguments holds for the effects.
\end{proof}

\begin{corollary}
    The CHSH value of the effective GPT of any instance of the group generator experiment satisfying Condition~\ref{cond:hierarchy_condition} is \emph{bounded above} by $2\sqrt2$, and hence (by Thm.~\ref{thm:noslerist_bound}) equal to $2\sqrt2$.
\end{corollary}

\begin{proof}
    This follows from the isomorphism to the qubit (Thm.~\ref{thm:isomorphic_to_qubit}) where Tsirelson's bound holds.
\end{proof}

%================================================================================================
% RELAXING THE IID ASSUMPTION
%================================================================================================

\section{Removing the IID assumption}\label{sec:relaxing_iid}

In this section, assuming we have access to an infinite chain of particles, we show that the IID assumption made at the beginning of Sec.~\ref{sec:teleportation_heirarchy_qubit} incurs no loss of generality.
To do so we use the de~Finetti representation theorem of \cite{christandltoner2009definetti}.

Consider a chain of particles labelled by the integers $\ZZ$. 
Pair the particles labelled by $(2k - 1, 2k)$.
For the $N$ level of the teleportation hierarchy, $\T_N$, instead of having access to $N+1$ identical copies of $\rho$, we now sample pairs $(2k - 1, 2k)$ from the infinite chain at random.
Demanding Condition~\ref{cond:hierarchy_condition} to hold for the state sampled in this way is equivalent to the infinite chain being invariant under permutation of the pairs $(2k - 1, 2k)$.
Therefore by the de~Finetti representation theorem of \cite{christandltoner2009definetti}, the state sampled in each $\T_N$ can be expressed as 
\begin{align*}
    \xi_{N+1} := \int dm(\lambda) \rho_\lambda^{\otimes N + 1},
\end{align*}
for some probability measure $m(\lambda)$ on an event space whose elements are indexed by $\lambda$, and bipartite states $\rho_\lambda$.

In principle, the states $\rho_\lambda$ should have an index $N$, since the IID decomposition for different $\xi_N$ can be different.
But, since we can use $\xi_N$ to perform every test $T_M$ for $M < N$, we can assume the $\rho_\lambda$ are same for every $N$.

We establish some notation required for the rest of this section.
Assume that the Erwins were given the settings $\vec \zeta \in \Z,\ \len(\vec \zeta) = N$, and obtained the outcome $\vec \mu$.
Let $\Cee^{(\vec \zeta)}_{\lambda, \vec \mu}$ denote the correlation matrix obtained if $\T_N$ is performed on $\rho_\lambda^{\otimes N+1}$.
Denote by $\Lambda^{(\vec \zeta)}_{\lambda, \vec \mu}$ the teleportation map for the IID state $\rho_\lambda^{\otimes N}$, i.e., 
\begin{align*}
    \Lambda^{(\vec \zeta)}_{\lambda, \vec \mu}
    =
    \prod_{i = \len(\vec \zeta)}^{1} \Lambda^{(z_i)}_{\lambda, \mu_i}
    =
    \prod_{i = \len(\vec \zeta)}^{1} \hat\rho_\lambda \hat\phi^{(z_i)}_{\mu_i}.
\end{align*}
Then the element of the state cone obtained in $\T_N$ after the Erwins have performed their measurement is
\begin{align*}
    \xi^{(\vec \zeta)}_{\vec \mu} 
    := 
    \int dm(\lambda) \Lambda^{(\vec \zeta)}_{\lambda, \vec \mu}\rho_\lambda.
\end{align*}
Correspondingly, the probability of obtaining $\vec \mu$ given $\vec \zeta$ is given by 
\begin{align*}
    p(\vec \mu | \vec \zeta) 
    = 
    \int dm(\lambda) p_\lambda(\vec \mu | \vec \zeta),
    \quad 
    p_\lambda(\vec \mu | \vec \zeta) := \Lambda^{(\vec \zeta)}_{\lambda, \vec \mu}\rho_\lambda(\1)(\1).
\end{align*}
Now, assuming Condition~\ref{cond:hierarchy_condition} we have 
\begin{align*}
    \xi^{(\vec \zeta)}_{\vec \mu}(A_\nu \otimes B_\alpha) 
    = 
    p(\vec \mu | \vec \zeta) \big(\Cee^{(\vec \zeta)}_{\vec \mu}\big)_{\nu\alpha}.
\end{align*}
By construction, 
\begin{align*}
    \Cee^{(\vec \zeta)}_{\vec \mu} = R\Em^{(0)},
\end{align*}
for some $R \in SO(3)$ (Sec.~\ref{subsec:quantum_realization_for_qubits} and Condition~\ref{cond:hierarchy_condition}).
Thus, it satisfies
\begin{equation}\label{eqn:orthogonality_C}
    \Cee^{(\vec \zeta)}_{\vec \mu} (\Cee^{(\vec \zeta)}_{\vec \mu})^T
    = R\Em^{(0)} (\Em^{(0)})^T R^T
    = \Id.
\end{equation}

\begin{definition}[Orthogonality test $\T^O_N$]
    Let $N \in \NN$ be fixed.
    Sample $\xi_{2N+2}$ from the infinite chain.
    Split it into two halves as follows 
    \begin{align*}
        \xi_{2N+2} 
        = \int dm(\lambda) \rho_\lambda^{\otimes N+1}\otimes\rho_\lambda^{\otimes N+1}.
    \end{align*}
    Perform the test $\T_N$ on each half, choosing the same setting for the Erwins as well as the Bobs, i.e., 
    measure 
    \begin{align*}
        \int dm(\lambda) \Lambda^{(\vec \zeta)}_{\lambda, \vec \mu}\rho_\lambda(A_\nu \otimes B_\beta) \Lambda^{(\vec \zeta)}_{\lambda, \vec \mu'}\rho_\lambda(A_\alpha \otimes B_\beta).
    \end{align*}
\end{definition}

The purpose of this test is to check a certain orthogonality condition on the matrices $\Cee^{(\vec \zeta)}_{\lambda, \vec \mu}$.
From the definition of $\xi^{(\vec \zeta)}_{\vec \mu}$ we have 
\begin{align*}
    \Cee^{(\vec \zeta)}_{\vec \mu} = \frac1{p(\vec \mu | \vec \zeta)} \int dm(\lambda)\ p_\lambda(\vec \mu | \vec \zeta) \Cee^{(\vec \zeta)}_{\lambda, \vec \mu}.
\end{align*}
On the other hand, the correlations from the orthogonality tests, for $\vec \mu = \vec \mu'$, allow us to measure 
\begin{align*}
    &
    \sum_\beta \int dm(\lambda)\ \Lambda^{(\vec \zeta)}_{\lambda, \vec \mu}\rho_\lambda(A_\nu \otimes B_\beta) \Lambda^{(\vec \zeta)}_{\lambda, \vec \mu}\rho_\lambda(A_\alpha \otimes B_\beta)
    \\
    &
    = \int dm(\lambda)\ p_\lambda(\vec \mu | \vec \zeta)^2 \Cee^{(\vec \zeta)}_{\lambda, \vec \mu}(\Cee^{(\vec \zeta)}_{\lambda, \vec \mu})^T.
\end{align*}

\begin{lemma}\label{lem:orth_to_iid}
    Let $\Cee \in \RR^{n \times n}$ be an orthogonal matrix ($\Cee\Cee^T = \Id$) admitting a convex decomposition
    \begin{align*}
        \Cee =  \int dm(\lambda)\ \Cee_\lambda,
    \end{align*}
    for $\Cee_\lambda \in \RR^{n \times n}$.
    Then $\forall \lambda,\ \Cee = \Cee_\lambda$ if and only if 
    \begin{align*}
        \int dm(\lambda)\ \Cee_\lambda\Cee_\lambda^T = \Id.
    \end{align*}
\end{lemma}

\begin{proof}
    The forward direction is trivial.
    Conversely, let 
    \begin{align*}
        \int dm(\lambda)\ \Cee_\lambda\Cee_\lambda^T = \Id.
    \end{align*}
    Consider 
    \begin{align*}
        &
        \int dm(\lambda)\ (\Cee_\lambda - \Cee)(\Cee_\lambda - \Cee)^T
        \\
        &
        = \int dm(\lambda)\ 
        (\Cee_\lambda\Cee_\lambda^T
        - \Cee\Cee_\lambda^T
        - \Cee_\lambda\Cee^T
        + \Cee\Cee^T)
        \\
        &
        = \Id
        - \Id
        - \Id
        + \Id = 0.
    \end{align*}
    Since $(\Cee_\lambda - \Cee)(\Cee_\lambda - \Cee)^T$ is positive semi-definite, this is only possible if
    \begin{align*}
        \forall \lambda,\ (\Cee_\lambda - \Cee)(\Cee_\lambda - \Cee)^T = 0.
    \end{align*}
    Since $M = \Cee_\lambda - \Cee$ is a real matrix, the $i^{\mathrm{th}}$ diagonal element of $MM^T$ is the euclidean norm of the $i^{\mathrm{th}}$ row of $M$.
    This is zero is and only if the $i^{\mathrm{th}}$ row of $M$ is the zero vector.
    Thus 
    \begin{align*}
        \forall \lambda,\ \Cee_\lambda - \Cee = 0.
    \end{align*}
\end{proof}

\begin{condition}
    For every $N \in \NN$, the correlation data from the orthogonality test $\T^O_N$ satisfies 
    \begin{align*}
        \tfrac1{p(\vec \mu | \vec \zeta)^2}\int dm(\lambda)\ p_\lambda(\vec \mu | \vec \zeta)^2 \Cee^{(\vec \zeta)}_{\lambda, \vec \mu}(\Cee^{(\vec \zeta)}_{\lambda, \vec \mu})^T =  \Id.
    \end{align*}
\end{condition}

Taking $\Cee = \Cee^{(\vec \zeta)}_{\vec \mu}$ and $\Cee_\lambda = \tfrac{p_\lambda(\vec \mu | \vec \zeta)}{p(\vec \mu | \vec \zeta)}\Cee^{(\vec \zeta)}_{\lambda, \vec \mu}$ in Lemma~\ref{lem:orth_to_iid} we get 
\begin{align*}
    \forall \lambda,\ \Cee^{(\vec \zeta)}_{\vec \mu} = \tfrac{p_\lambda(\vec \mu | \vec \zeta)}{p(\vec \mu | \vec \zeta)}\Cee^{(\vec \zeta)}_{\lambda, \vec \mu}.
\end{align*}
This implies that the IID states $\rho_\lambda^{\otimes N+1}$ all reproduce the correlations of $\T_{N}$ demanded in Condition~\ref{cond:hierarchy_condition}.
Thus, without loss of generality we can assume $\xi_N = \rho^{\otimes N}$.

\begin{remark}
    Since each IID state reproduces the correlation of Condition~\ref{cond:hierarchy_condition}, Lemma~\ref{lem:prob_quarter_all} applies, and all the outcomes $\vec \mu$ must be equiprobable.
\end{remark}

%================================================================================================
% GENERALIZATIONS
%================================================================================================

\section{Obtaining the rest finite dimensional quantum theory}\label{sec:fin_dim_qm}

What remains to be shown is that the procedure to characterize the qubit can be generalized to $(\CC^d)^{\otimes n}$, i.e., $n$-qudits, for every $d,n \in \NN$.

The core idea is the same. 
We use the teleportation protocol with qudit Paulis as the correction group to certify the $n$-qudit Pauli ``observables''.
These observables form a tomographically complete set on the space of operators on $\CC^d$.
This allows us to encode the matrix elements of generators of the symmetry group of the $n$-qudit state space into the correlations as in the qubit case, and subsequently extract the generators as teleportation maps.

We split the exposition into two parts. 
First we will generalize from one to $n$ qubits, and then, separately, from qubits to qudits.

%========================================================
% Pauli operators
%========================================================

\subsection{Some notation for the Pauli operators}\label{subsec:paulis}

The Pauli operators and their qudit generalizations will play a cental role in this section.
Here we establish the required notation.

Let $\{ |q\rangle \}_{q = 0}^{d-1} \subseteq \CC^d$ be an orthonormal basis.
Let $\omega := e^{i \frac{2\pi}d}$.
The \emph{(generalized) Pauli operators} are defined by their action on such a basis
\begin{align*}
    Z |q\rangle = \omega^{q} |q\rangle, \qquad 
    X |q\rangle = |q \oplus_d 1\rangle.
\end{align*}
The corresponding \emph{Weyl operators} on $(\CC^d)^{\otimes n}$ are defined as follows:
\begin{equation}\label{eqn:weyl}
    \forall (\vec p, \vec q) \in \ZZ_d^n \times \ZZ_d^n,\ w_{\vec p, \vec q} := \bigotimes_{j = 1}^n Z^{p_j}X^{q_j}.
\end{equation}

One important distinction between the qubit and qudit Pauli operators, is that the qudit versions are \emph{not hermitian} for $d>2$.
But we will none-the-less be ``measuring'' these operators.
We will do this by measuring in the PVM given by the eigenbasis of the Pauli operator, and then use the relative frequencies as weights for the appropriate phases to get the expected value.
Concretely, let $A_s$ and $B_t$ be qudit Pauli operators, and let $\{ P^{(s)}_a \}_{a \in \ZZ_d}$ and $\{ Q^{(t)}_b \}_{b \in \ZZ_d}$ be their respective eigenbasis PVMs.
Let the eigenvalues corresponding to the projectors $P^{(s)}_a$ and $Q^{(t)}_b$ be $\alpha_a$ and $\beta_b$ respectively, i.e., 
\begin{align*}
    \tr(A_s P^{(s)}_a) = \alpha_a, \qquad 
    \tr(B_t Q^{(t)}_b) = \beta_b.
\end{align*}
Then the expected values of $A_s \otimes B_t$ with respect to the state $|\psi\rangle$ are calculated as 
\begin{equation}\label{eqn:generalized_pauli_expt_vals}
    \langle A_s \otimes B_t \rangle = \sum_{a, b\in \ZZ_d} \alpha_a \beta_b \langle \psi | P^{(s)}_a \otimes Q^{(t)}_b | \psi \rangle,
\end{equation}
which is in general complex valued.

%========================================================
% n qubits
%========================================================

\subsection{The Group Generator Experiment for $n$ qubits}\label{subsec:group_gen_expt_n_qubits}

For the experiment for $n$ qubits we build a star network with $2n$ branches, with the branches divided into two sets of $n$.
An example for $n = 3$ is depicted in Fig.~\ref{fig:n_qubits}.
We will call the cental note of the network Erwin${}^\star$.

\begin{figure}[H]
    \centering
    \includegraphics[trim={0cm 3.5cm 0cm 3.5cm}, clip, width = \linewidth, keepaspectratio]{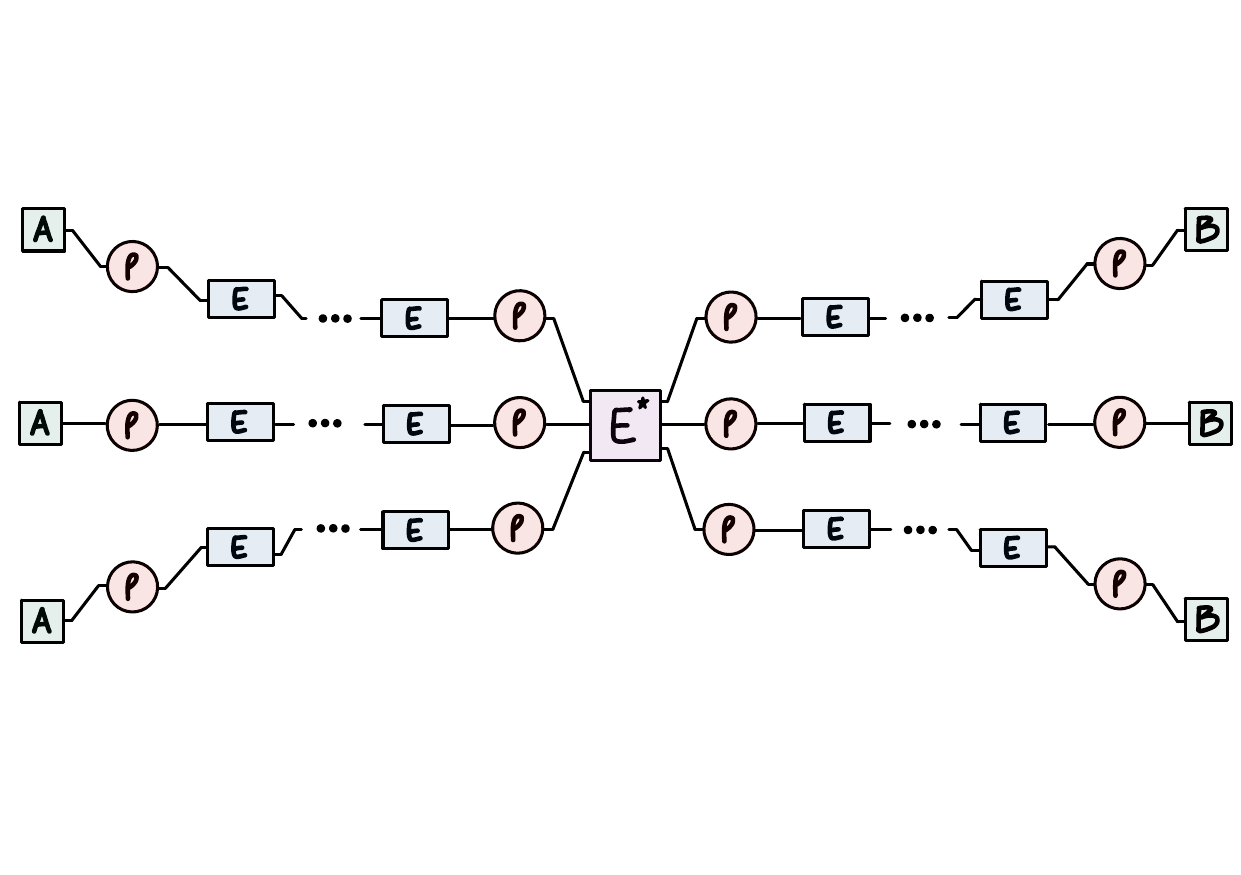}
    \caption{The setup for the teleportation hierarchy for three qubits}
    \label{fig:n_qubits}
\end{figure}

The main idea here is to first treat the bipartite star as $n$ independent chains, and independent certify qubit Pauli observables on each chain.
Then, on the central node, we couple different chains to test for the generators of the $n$-qubit symmetry group.

%========================================================
% n-qubits qm strategy
%========================================================

\subsubsection{Quantum realization}\label{subsec:quantum_realization_for_n_qubits}

For the quantum realization of the $n$ qubit group generator experiment, consider the a star graph of the kind Fig.~\ref{fig:n_qubits}, but with just Erwin${}^\star$ and no Erwins.
That is $n$ Alices, $n$ Bobs, one Erwin${}^\star$ and $2n$ copies of $\rho$.
The measurement of the Alices and Bobs, and all the states are identical to the single qubit version (Sec.~\ref{subsec:quantum_realization_for_qubits}).

Let $\xi_r$ denote the identity, and $\{\xi_r\}_r \subseteq U(2^n)$ be such that $\{ \xi_r (\cdot) \xi_r^\dag \}_r$ generate the entire symmetry group of the $n$-qubit state space.
Then the measurement basis for Erwin${}^\star$ given setting $r$ is defined by 
\begin{align*}
    \{ (w_{\vec p, \vec q}\xi_r^\dag)_{1, \dots, n} \otimes \Id_{n+1, \dots, 2n} \bigotimes_{i = 1}^n | \Phi^+ \rangle_{i, n+i} \ | \ \vec p, \vec q \in \ZZ_2^n \}.
\end{align*}

As stated, the number of outcomes of Erwin${}^\star$ and the size of the correlation matrices grow exponentially in $n$.
But a smart choice of generators mitigates this issue.
Let $\{ \eta_k \}_{k = 1,2} \subseteq U(2^2)$ be such that $\{ \eta_k (\cdot) \eta_k^\dag \}_{k = 1,2}$ generate the full symmetry group of the two-qubit state space.
Let $\eta_k^{(i, i+1)}$ denote that the operator acts on subsystems $i$ and $i+1$, i.e., 
\begin{align*}
    \eta_k^{(i, i+1)} = \Id^{\otimes i - 1} \otimes \eta_k \otimes \Id^{\otimes n - (i + 1)}.
\end{align*}
Then, it is know that the operators 
\begin{align*}
    \{\eta_k^{(i, i+1)}\ | \ k = 1,2 \ \text{and} \ i = 1, \dots, n-1  \}
\end{align*}
generate $U(2^n)$ (direct consequence of \cite{DiVincenzo1995twoqubituniversal} and \cite{breuillard2002densefreesubgroupslie}).
Thus we can separate the entire structure into $n-1$ experiment.
In experiment $i$, the cental node performs the two-qubit experiment for $A^{(i)}, A^{(i+1)}, B^{(i)}, B^{(i+1)}$ ($A^{(i)} \equiv$ the $i^{\mathrm{th}}$ Alice).
The example for three qubits is illustrated in Fig.~\ref{fig:n_qubits_n-1_expts}.
This in turn implies that the correlation matrix have constant size ($16 \times 16$) and the number of outcomes of Erwin${}^\star$ are at most $16$.

\begin{figure}[H]
    \centering
    \includegraphics[trim={0cm 3.5cm 0cm 3.5cm}, clip, width = 0.8\linewidth, keepaspectratio]{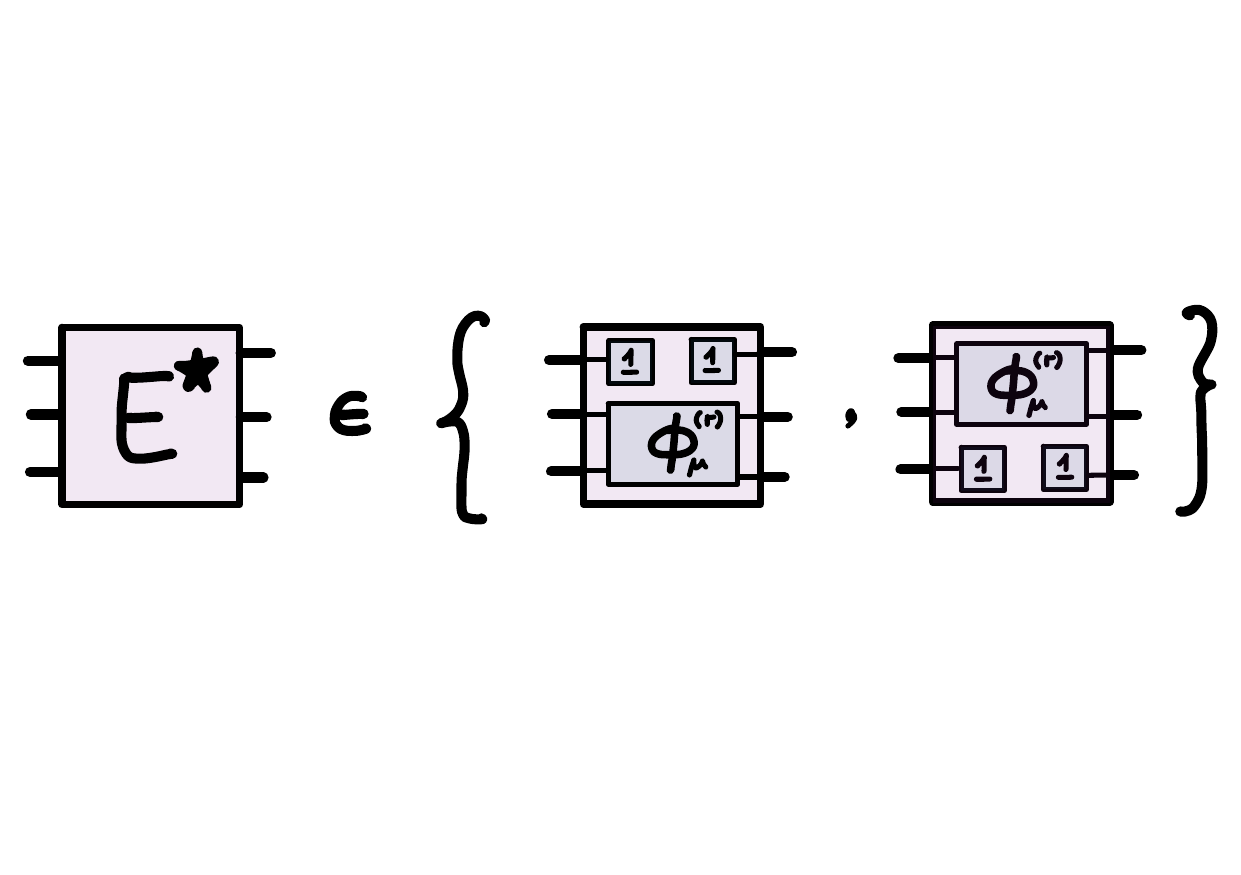}
    \caption{The two-qubit experiments sufficient for the three qubit case.}
    \label{fig:n_qubits_n-1_expts}
\end{figure}

%========================================================
% n-qubits teleportation hierarchy
%========================================================

\subsubsection{The teleportation hierarchy for $n$ qubits}\label{sec:teleportation_heirarchy_n_qubit}

The teleportation hierarchy for $n$ qubits is near identical to that of $1$ qubit, with one difference -- only Erwin${}^\star$ has the freedom to choose settings other that $0$.
The graphical structure of the hierarchy is as depicted in Fig.~\ref{fig:n_qubits}.
The condition also has the same structure Condition~\ref{cond:hierarchy_condition}.
The sixteen dimensional correlation matrices $\Cee^{(r)}_\mu$ are calculated from the quantum realization (Sec.~\ref{subsec:quantum_realization_for_n_qubits}) as it was in the case of one qubit.

%========================================================
% qudits
%========================================================

\subsection{The Group Generator Experiment for qudits}\label{subsec:group_gen_expt_qudits}

The setup for the qudit case is identical to the qubit case. 
What changes is the number of settings of Alice and Bob, and the number of outcomes of Alice, Bob and Erwin.
The setup is illustrated in Fig.~\ref{fig:qudit_gg_test}

\begin{figure}[H]
    \centering
    \includegraphics[trim={0cm 3.5cm 0cm 3.5cm}, clip, width = \linewidth, keepaspectratio]{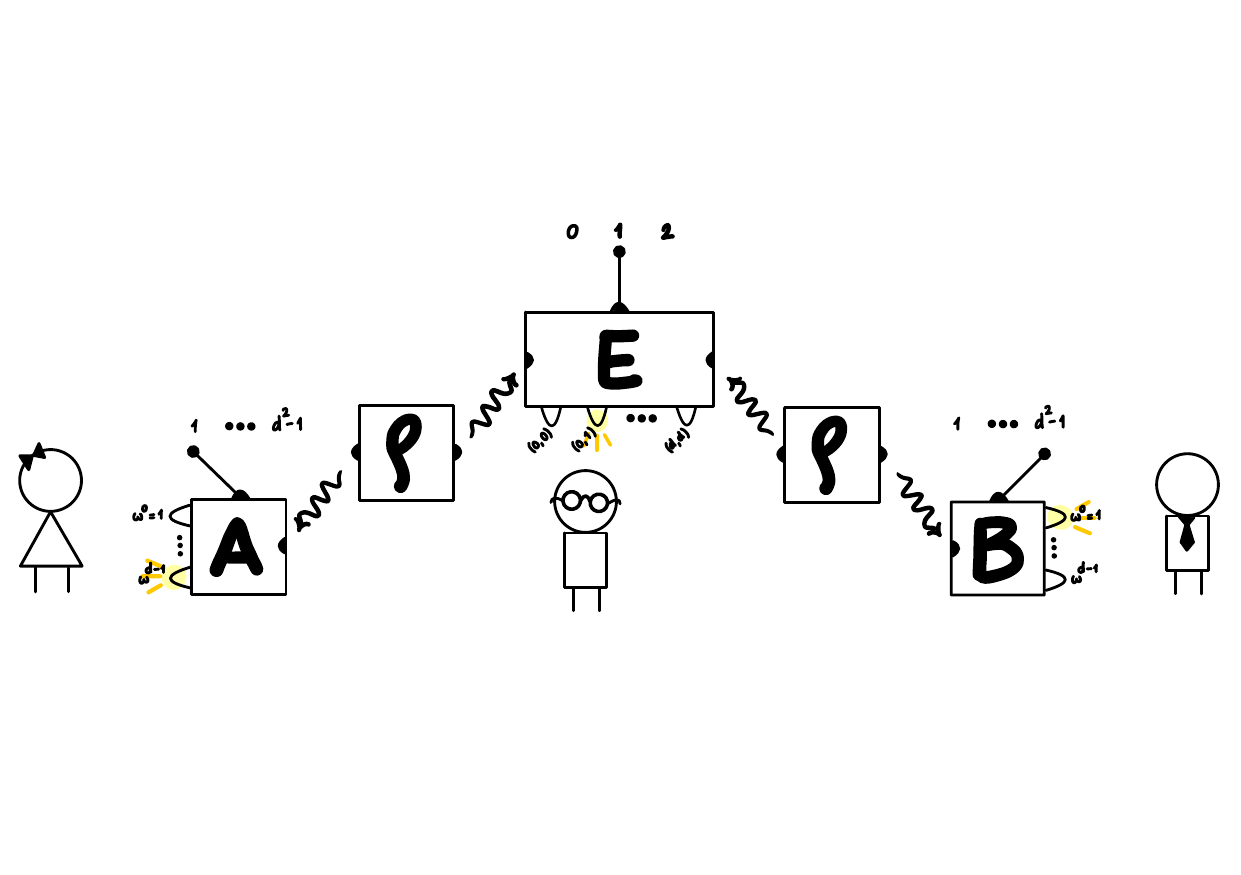}
    \caption{The entanglement swapping experiment for the qudit.}
    \label{fig:qudit_gg_test}
\end{figure}

In this case Erwin has $d^2$ outcomes, one corresponding to each generalized pauli operator (including identity).
Alice and Bob have $d^2 - 1$ settings, (Paulis minus identity) and $d$ outcomes per setting.
As discussed in Sec.~\ref{subsec:paulis}, the expectation values in this case will be calculated using Eqn.~(\ref{eqn:generalized_pauli_expt_vals}).

%========================================================
% qudits qm strategy
%========================================================

\subsubsection{Quantum realization}\label{subsec:quantum_realization_for_qudits}

Let $d \in \NN,\ d \geq 2$ be fixed.
Let $\{|k\rangle\}_{k \in \ZZ_d}$ be an ONB of $\CC^d$.
Let the bipartite state shared between Alice and Erwin, and Erwin and Bob be
\begin{align*}
    |\Phi_d^+\rangle := \tfrac1{\sqrt d} \sum_{k = 0}^{d-1} | kk\rangle.
\end{align*}
For index set $I := \ZZ_d^2 \setminus \{ (0,0) \}$, the observables of Alice and Bob are given 
\begin{align*}
    \forall (p,q) \in I,\ A_{p,q} = B_{p,q} = w_{p,q}.
\end{align*}
Let $\xi_0 = \Id$, and $\{ \xi_r \}_{r = 1,2} \subseteq U(d)$ be such that $\{ \xi_r (\cdot) \xi_r^\dag \}_{r = 1,2}$ generates the symmetry group of the qudit state space.
Then, for $r \in \ZZ_3$, the measurement bases of Erwin are given by 
\begin{align*}
    \{ w_{p,q}\xi_r^\dag \otimes \Id |\Phi_d^+\rangle \ | \ (p,q) \in \ZZ_d \times \ZZ_d \}.
\end{align*}

%========================================================
% qudits teleportation hierarchy
%========================================================

\subsubsection{The teleportation hierarchy for qudits}\label{sec:teleportation_heirarchy_qudit}

The teleportation hierarchy is once again structurally identical to the qubit case.

One of the core differences between the qubit and the qudit case is the correction group.
The adjoint action of the qudit Paulis correspond to the regular representation of the group $\ZZ_d \times \ZZ_d$ under vector addition modulo $d$.
The fact that it is regular representation leads to a straightforward generalization of the equiprobability lemma, Lem.~\ref{lem:prob_quarter_all}

\begin{lemma}\label{lem:qudit_prob_quarter_all}
    For the qudit setting ($d \in \NN$), for every $\vec \zeta \in \Z$, and every $\vec \mu \in (\ZZ_d^2)^{\len(\vec \zeta)}$, 
    \begin{align*}
        p(\vec \mu | \vec \zeta) = \big(\tfrac1{d^2}\big)^{\len(\vec \zeta)}.
    \end{align*}
\end{lemma}

This allows us to re-derive all the results of Sec.~\ref{subsec:comp_topo_semigrp} for qudits. 
This is because we never explicitly use the properties of the underlying group in Sec.~\ref{subsec:comp_topo_semigrp}.
The only result in Sec.~\ref{subsec:comp_topo_semigrp} that uses the properties of $\ZZ_2^2$ is Lem.~\ref{lem:reg_rep_proof}.
But the conditions of this lemma are also satisfied by $\ZZ_d^2$.

%================================================================================================
% CONCLUSIONS AND OUTLOOK
%================================================================================================

\section{Conclusions and Outlook}\label{sec:conclusion_and_outlook}

In this work, we have shown that the correlations of the group generator experiment (Secs.~\ref{sec:group_gen_expt_qubit},~\ref{subsec:group_gen_expt_n_qubits},~and~\ref{subsec:group_gen_expt_qudits}), together with the assumption that these correlations are stable under iterated teleportation, fully characterize finite-dimensional quantum theory.
More precisely, we construct a hierarchy of entanglement swapping experiments that is complete for $n$ quantum systems of dimension $d$ each.
One main outstanding question is whether, and/or under what assumptions can this hierarchy be truncated at a finite level. 
This opens up the possibility of a finitary realization of the characterization.

% For each $(\CC^d)^{\otimes n}$ we have constructed a hierarchy of entanglement swapping experiments that can certify the existence of all possible quantum correlation realizable with $(\CC^d)^{\otimes n}$.
% In particular, our hierarchy can certify the existence of correlation that have not explicitly been observed in experiment.
% Thus, the realizability of the group generator correlations Sec.~\ref{sec:group_gen_expt_qubit} (Sec.~\ref{subsec:group_gen_expt_n_qubits}, Sec.~\ref{subsec:group_gen_expt_qudits}) together with the assumption that these correlation are preserved under iterated teleportation, fully characterize finite dimensional quantum theory.

% Our results, as stated, require that we test an infinite set of correlation.
% One future direction of research is to ask whether it is truly necessary to have infinitely many tests, or, can we meaningfully draw the same conclusion from finitely many correlations.
% If it is indeed the case that finitely many steps are enough, this result can be converted into a proper self-test for finite dimensional quantum theory, with the possibility of a laboratory realization.

%================================================================================================
% ACKNOWLEDGEMENTS
%================================================================================================

\section{Acknowledgements}\label{sec:acknowledgments}

We are incredibly grateful to Johan {\AA}berg for all the discussions and feedback throughout this work.
We are also grateful to Sadra Boreiri, Nikolai Miklin, Pavel Sekatski, and Lucas Vieira for their helpful discussions.

LD and DG acknowledge support from 
Germany's Excellence Strategy -- Cluster of Excellence Matter and Light for Quantum Computing
(ML4Q) EXC 2004/2 (390534769).
XX and MOR acknowledge funding by the ANR for the JCJC grants LINKS (No. ANR-23-CE47-0003), the T-ERC QNET (No. ANR-24-ERCS-0008), and the ANR-NSF PRCI grant QUANTINT (No. ANR-25-CE48-3338). 
All authors acknowledge support from QuantERA II Programme that has received funding from the EU's H2020 research and innovation programme under the GA No.~101017733 COMPUTE.

\subsection{Artificial Intelligence usage declaration}

Artificial intelligence was \emph{not} used in any meaningful capacity in this work.

%%%%%%%%%%%%%%%%%%%%%%%%%%%%%%%%%%%%%%%%%%%%%%%%%%%%%%%%%%%%%%%%%%%%%%%%%%%%%%%%%%%%%%%%%%%%%%%%%%%%%%%%%%%%%%%%%%%%%%%%%%%%%%%%%

\bibliography{bibliography}

%%%%%%%%%%%%%%%%%%%%%%%%%%%%%%%%%%%%%%%%%%%%%%%%%%%%%%%%%%%%%%%%%%%%%%%%%%%%%%%%%%%%%%%%%%%%%%%%%%%%%%%%%%%%%%%%%%%%%%%%%%%%%%%%%

\onecolumngrid

\section{Appendix}

\subsection{Derivation of correlation matrices in the QM case}\label{app:qm_corr_mats}

Define 
\begin{align*}
    |\Phi^+_{r, \mu}\rangle\langle\Phi^+_{r, \mu}| := \sigma_\mu\xi_r^\dag \otimes \Id |\Phi^+\rangle\langle\Phi^+| (\sigma_\mu\xi_r^\dag)^\dag \otimes \Id.
\end{align*}
Given that Erwin obtained the outcome $\mu$ for setting $r$, we can calculate the state after entanglement swapping by using Eqn.~(58) from \cite{dmello2024entanglement} which reads as follows:
\begin{align*}
    & \tr_{23}\big(
    (A^\dag \otimes \Id)_{12}
    |\Phi^+ \rangle \langle \Phi^+ |_{12} 
    (A \otimes \Id)_{12}
    (B^\dag \otimes \Id)_{23}
    |\Phi^+ \rangle \langle \Phi^+ |_{23}
    (B \otimes \Id)_{23}
    (C^\dag \otimes \Id)_{34}
    |\Phi^+ \rangle \langle \Phi^+ |_{34}
    (C \otimes \Id)_{34}
    \big)
    \\
    &=
    \tfrac14 (C \overline{B}A)^\dag \otimes \Id |\Phi^+ \rangle \langle \Phi^+ | (C \overline{B}A) \otimes \Id.
\end{align*}
In our case, $A = C = \Id$, and $B = (\sigma_\mu \xi_r^\dag)^\dag = \xi_r \sigma_\mu$
Thus with probability $1/4$ (for each $\mu$), we get
\begin{align*}
    |\Psi_{r, \mu}\rangle\langle\Psi_{r, \mu}|
    :=
    \swap{|\Phi^+\rangle\langle\Phi^+|}{|\Phi^+_{r, \mu}\rangle\langle\Phi^+_{r, \mu}|}{|\Phi^+\rangle\langle\Phi^+|}
    = ( \overline{\xi_r\sigma_\mu})^\dag \otimes \Id |\Phi^+ \rangle \langle \Phi^+ | (\overline{\xi_r\sigma_\mu}) \otimes \Id.
\end{align*}
The unitary $\xi_r$ are nothing but
\begin{align*}
    \xi_0 = \Id, \qquad
    \xi_1 = e^{i\frac\theta2 \sigma_z}, \qquad
    \xi_2 = e^{i\frac\theta2 \sigma_x}.
\end{align*}
Thus 
\begin{align*}
    \overline{\xi_0} = \xi_0 = \xi_0^\dag, \qquad
    \overline{\xi_1} = \xi_1^\dag, \qquad
    \overline{\xi_2} = \xi_2^\dag,
\end{align*}
which gives 
\begin{align*}
    |\Psi_{r, \mu}\rangle\langle\Psi_{r, \mu}|
    = (\sigma_\mu\xi_r) \otimes \Id |\Phi^+ \rangle \langle \Phi^+ | (\xi_r^\dag\sigma_\mu) \otimes \Id.
\end{align*}
For this part we will use the convention $A_\mu = B_\mu = \sigma_\mu$ with the standard notation for the Pauli matrices.
It differs from the convention in the main body by a re-ordering of blocks.
Define the function
\begin{align*}
    s(\mu, \nu) 
    =
    \begin{cases}
        0 & \ \text{if} \ \mu = 0 \ \text{or} \ \nu = 0 \ \text{or} \ \mu = \nu \\
        1 & \ \text{else}.
    \end{cases},
\end{align*}
and calculate the correlation function 
\begin{align*}
    \Big(\Cee^{(r)}_{\mu}\Big)_{\nu\alpha}
    =
    \langle A_\nu B_\alpha\rangle_{r, \mu}
    &
    =
    \langle \Psi_{r, \mu}| A_\nu \otimes B_\alpha |\Psi_{r,\mu} \rangle
    \\
    &
    =
    \langle \Phi^+ | (\xi_r^\dag\sigma_\mu A_\nu \sigma_\mu\xi_r) \otimes B_\alpha  |\Phi^+ \rangle
    \\
    &
    =
    (-1)^{s(\mu, \nu)} \langle \Phi^+ | (\xi_r^\dag A_\nu \xi_r) \otimes B_\alpha  |\Phi^+ \rangle
\end{align*}
Now we use the following 
\begin{align*}
    \langle\Phi^+| \sigma_\nu \otimes \sigma_\alpha |\Phi^+\rangle = (-1)^{\delta_{2\nu}} \delta_{\nu\alpha}.
\end{align*}
First we can treat the two extreme cases.
If $\nu = 0$, then 
\begin{align*}
    \Big(\Cee^{(r)}_{\mu}\Big)_{0\alpha} 
    =
    (-1)^{s(\mu, 0)} \langle \Phi^+ | (\xi_r^\dag A_0 \xi_r) \otimes B_\alpha  |\Phi^+ \rangle
    =
    \langle \Phi^+ | \Id \otimes B_\alpha  |\Phi^+ \rangle
    =
    \delta_{0\alpha}
\end{align*}
Similarly, if $\alpha = 0$ then $\Big(\Cee^{(r)}_{\mu}\Big)_{\nu0} = \delta_{\nu0}$.

Assume $\nu \neq 0$ and $\alpha \neq 0$.
Then, 
\begin{itemize}
    \item 
    For $r = 0$ we have 

    \begin{align*}
        \Big(\Cee^{(0)}_{\mu}\Big)_{\nu\alpha} 
        =
        (-1)^{s(\mu, \nu)} \langle \Phi^+ | A_\nu \otimes B_\alpha  |\Phi^+ \rangle 
        =
        (-1)^{s(\mu, \nu)} (-1)^{\delta_{2\nu}} \delta_{\nu\alpha} 
    \end{align*}

    \item
    For $r = 1$ the rotations give 
    \begin{align*}
        \xi_1^\dag A_1 \xi_1 = A'_1 = \cos\theta A_1 + \sin\theta A_2, \quad
        \xi_1^\dag A_2 \xi_1 = A'_2 = -\sin\theta A_1 + \cos\theta A_2, \quad
        \xi_1^\dag A_3 \xi_1 = A'_3 = A_3.
    \end{align*}
    Therefore 
    \begin{align*}
        \Big(\Cee^{(1)}_{\mu}\Big)_{\nu\alpha}
        =
        (-1)^{s(\mu, \nu)} \langle \Phi^+ | A'_\nu \otimes B_\alpha  |\Phi^+ \rangle
    \end{align*}

    \item
    For $r = 2$ the rotations give 
    \begin{align*}
        \xi_2^\dag A_1 \xi_2 = A''_1 = A_1, \quad
        \xi_2^\dag A_2 \xi_2 = A''_2 = \cos\theta A_2 + \sin\theta A_3, \quad
        \xi_2^\dag A_3 \xi_2 = A''_3 = -\sin\theta A_2 + \cos\theta A_3.
    \end{align*}
    Therefore 
    \begin{align*}
        \Big(\Cee^{(2)}_{\mu}\Big)_{\nu\alpha}
        =
        (-1)^{s(\mu, \nu)} \langle \Phi^+ | A''_\nu \otimes B_\alpha  |\Phi^+ \rangle
    \end{align*}
\end{itemize}

%%%%%%%%%%%%%%%%%%%%%%%%%%%%%%%%%%%%%%%%%%%%%%%%%%%%%%%%%%%%%%%%%%%%%%%%%%%%%%%%%%%%%%%%%%%%%%%%%%%%%%%%%%%%%%%%%%%%%%%%%%%%%%%%%

\subsection{Compactness of GPT state spaces}\label{app:compact_state_sapce}

\begin{lemma*}
    The state space $\S^{(n)}$ is a compact set.
\end{lemma*}

\begin{proof}
    The definition of $\D^{(n)}$ is such that the cone lies entirely in $\Span(\P^{(n)})^*$, where it is a closed, pointed, generating, convex cone (see Sec.~\ref{subsec:GPTs}).
    Thus it is sufficient for the sake of this proof to restrict the ambient spaces of $\D^{(n)}$ and $\P^{(n)}$ to be $\Span(\P^{(n)})^*$ and $\Span(\P^{(n)})$ respectively, and thus assume $\D^{(n)}$ and $\P^{(n)}$ to be closed, pointed, generating, convex cones (see footnotes in Sec.~\ref{subsec:GPTs} for definitions).

    Since we are in finite dimensions it is sufficient to show that $\S^{(n)}$ is closed and bounded.
    The cone $\D^{(n)}$ is closed by definition (Sec.~\ref{subsec:GPTs}) and $\1^{\otimes n}$ is a continuous linear functional.
    $\S^{(n)}$ arises as the intersection of $\D^{(n)}$ and the pre-image of $1$ under $\1^{\otimes n}$ (c.f.\ definition in Sec.~\ref{subsec:GPTs}) and thus is closed.

    By definition $\1^{\otimes n}$ is in the relative interior of $\P^{(n)} \subset (\D^{(n)})'$ (the polar dual of $\D^{(n)}$, see footnote in Sec.~\ref{subsec:GPTs}).
    Therefore, taking the ambient space to be $\Span(\P^{(n)})$ implies that $\1^{\otimes n}$ is in fact in the interior of $\P^{(n)}$.
    As a consequence, it is a strictly positive functional on $\D^{(n)}$, i.e., it is positive on $\D^{(n)} \setminus \{0\}$.
    Assume for the sake of reaching a contradiction that $\1^{\otimes n}$ is positive but not strictly positive on $\D^{(n)}$.
    Then, there exists some $\sigma \in \D^{(n)}$ such that $\sigma(\1^{\otimes n}) = 0$.
    Since $\1^{\otimes n}$ is a continuous functional in the interior of $(\D^{(n)})'$, there exists some perturbation of $p$ of $\1^{\otimes n}$ in the interior of $(\D^{(n)})'$ such that $\sigma(p) < 0$.
    This contradicts the fact that $p \in (\D^{(n)})'$.
    
    Now, let $S_1$ be the unit sphere with respect to the Euclidean norm.
    Because $\D^{(n)}$ is closed, the intersection $\D^{(n)} \cap S_1$ is a compact set.
    Let $\lambda_{\mathrm{min}} > 0$ be the minimum value attained by $\1^{\otimes n}$ on $\D^{(n)} \cap S_1$.
    Then the norm ball $B_{1/\lambda_{\mathrm{min}}}$ of radius $1/\lambda_{\mathrm{min}}$ contains $\S^{(n)}$.
    Thus $\S^{(n)}$ is compact.
\end{proof}

%%%%%%%%%%%%%%%%%%%%%%%%%%%%%%%%%%%%%%%%%%%%%%%%%%%%%%%%%%%%%%%%%%%%%%%%%%%%%%%%%%%%%%%%%%%%%%%%%%%%%%%%%%%%%%%%%%%%%%%%%%%%%%%%%

\end{document}